\documentclass[onefignum,onetabnum]{siamonline220329}

\makeatletter
\newenvironment{breakablealgorithm}
  {
   \begin{center}
     \refstepcounter{algorithm}
     \hrule height.8pt depth0pt \kern2pt
     \renewcommand{\caption}[2][\relax]{
       {\raggedright\textbf{\ALG@name~\thealgorithm} ##2\par}%
       \ifx\relax##1\relax 
         \addcontentsline{loa}{algorithm}{\protect\numberline{\thealgorithm}##2}%
       \else 
         \addcontentsline{loa}{algorithm}{\protect\numberline{\thealgorithm}##1}%
       \fi
       \kern2pt\hrule\kern2pt
     }
  }{
     \kern2pt\hrule\relax
   \end{center}
  }
\makeatother
\usepackage{booktabs, tabularx} 
\usepackage[export]{adjustbox}
\usepackage{listings}
\usepackage[all]{xypic}
\usepackage{lipsum}
\usepackage{amsfonts,amssymb}
\usepackage{graphicx}
\usepackage{epstopdf}
\usepackage{mathtools}
\usepackage{algorithmic}
\usepackage{booktabs}
\usepackage{subcaption}
\usepackage{cancel}
\usepackage{graphbox}
\ifpdf
\DeclareGraphicsExtensions{.eps,.pdf,.png,.jpg}
\else
  \DeclareGraphicsExtensions{.eps}
\fi

\usepackage{enumitem}
\setlist[enumerate]{leftmargin=.5in}
\setlist[itemize]{leftmargin=.5in}

\DeclareMathOperator{\im}{im}

\newcommand{\sett}[1]{\left\{#1\right\}}

\newcommand{\RR}{\mathbb{R}}
\newcommand{\CC}{\mathbb{C}}

\newcommand{\snake}[1]{\text{\url{#1}}}

\newcommand{\norm}[1]{||#1||}
\newcommand{\rank}{\text{rank}}
\newcommand{\tr}{\mathrm{tr}}

\renewcommand{\epsilon}{\varepsilon}

\newsiamremark{remark}{Remark}
\newsiamremark{hypothesis}{Hypothesis}
\crefname{hypothesis}{Hypothesis}{Hypotheses}
\newsiamthm{claim}{Claim}
\newsiamthm{assumption}{Assumption}

\title{$O(k)$-Equivariant Dimensionality Reduction on Stiefel Manifolds}

\author{
Andrew Lee\thanks{St. Thomas Aquinas College, Sparkill, NY (\email{alee@stac.edu}).}
\and Harlin Lee\thanks{University of North Carolina at Chapel Hill, Chapel Hill, NC 
(\email{harlin@unc.edu}).}
\and Jose A. Perea \thanks{Northeastern University, Boston, MA (\email{j.pereabenitez@northeastern.edu}).}
\and Nikolas Schonsheck\thanks{The Rockefeller University, New York, NY (\email{nschonsheck@rockefeller.edu}). Corresponding author.}
\and Madeleine Weinstein\thanks{University of Puget Sound, Tacoma, WA (\email{mweinstein@pugetsound.edu}).}
}
\usepackage{amsopn}

\definecolor{codegreen}{rgb}{0,0.6,0}
\definecolor{codegray}{rgb}{0.5,0.5,0.5}
\definecolor{codepurple}{rgb}{0.58,0,0.82}
\definecolor{backcolour}{rgb}{0.95,0.95,0.92}

\lstdefinestyle{mystyle}{
    backgroundcolor=\color{backcolour},   
    commentstyle=\color{codegreen},
    keywordstyle=\color{magenta},
    numberstyle=\tiny\color{codegray},
    stringstyle=\color{codepurple},
    basicstyle=\ttfamily\footnotesize,
    breakatwhitespace=false,         
    breaklines=true,                 
    captionpos=b,                    
    keepspaces=true,                 
    numbersep=5pt,                  
    showspaces=false,                
    showstringspaces=false,
    showtabs=false,                  
    tabsize=2
}

\begin{document}

\maketitle

\begin{abstract}
Many real-world datasets live on high-dimensional Stiefel and Grassmannian manifolds, $V_k(\mathbb{R}^N)$ and $Gr(k, \mathbb{R}^N)$ respectively, and benefit from projection onto lower-dimensional Stiefel and Grassmannian manifolds. In this work, we propose an algorithm called \textit{Principal Stiefel Coordinates (PSC)} to reduce data dimensionality from $ V_k(\mathbb{R}^N)$ to $V_k(\mathbb{R}^n)$ in an \textit{$O(k)$-equivariant} manner ($k \leq n \ll N$). We begin by observing that each element $\alpha \in V_n(\mathbb{R}^N)$ defines an isometric embedding of $V_k(\RR^n)$ into $V_k(\RR^N)$. Next, we describe two ways of finding a suitable embedding map $\alpha$: one via an extension of principal component analysis ($\alpha_{PCA}$), and one that further minimizes data fit error using gradient descent ($\alpha_{GD}$). Then, we define a continuous and $O(k)$-equivariant map $\pi_\alpha$ that acts as a ``closest point operator'' to project the data onto the image of $V_k(\RR^n)$ in $V_k(\mathbb{R}^N)$ under the embedding determined by $\alpha$, while minimizing distortion. Because this dimensionality reduction is $O(k)$-equivariant, these results extend to Grassmannian manifolds as well. Lastly, we show that  $\pi_{\alpha_{PCA}}$ globally minimizes projection error in a noiseless setting, while $\pi_{\alpha_{GD}}$ achieves a meaningfully different and improved outcome when the data does not lie exactly on the image of a linearly embedded lower-dimensional Stiefel manifold as above. Multiple numerical experiments using synthetic and real-world data are performed.
\end{abstract}



\section{Introduction}
Stiefel manifolds are a natural setting for data sets appearing in a wide variety of disciplines, including image processing~\cite{turaga2011statistical}, sensor array processing~\cite{ramirez2022subspace}, bioinformatics~\cite{tian2021clustering}, and astronomy~\cite{jupp_mardia_directional_statistics}. In recent years, there has been broad interest in dimensionality reduction methods for manifold-valued data \cite{chami2021horoPCA,sungkyudryden2012nestedspheres,melba:2022:002:yang}. The work of \cite{perea2018multiscale} introduced Principal Projective Component Analysis, an adaptation of Principal Component Analysis (PCA) to the settings where $O(1)$ acts on $V_1(\mathbb{R}^n)$ by scalar multiplication (with quotient equal to the real projective space $\mathbb{R}\textbf{P}^{n-1}$) and where the unitary group $U(1)$ acts on $V_1(\mathbb{C}^n)$ via scalar multiplication (with quotient equal to the complex projective space $\mathbb{C}\textbf{P}^{n-1}$). In a similar vein, \cite{Perea_Polanco}  introduced Lens PCA, a modification of PCA suitable for lens spaces, which are the quotient spaces corresponding to the action of $\mathbb{Z}_q$ on $V_1(\mathbb{C}^n)$ given by scalar multiplication by powers of a primitive $q$-th root of unity. Generalizing methods such as PCA or multidimensional scaling to non-Euclidean spaces brings familiar tools to bear on these problems, while also incorporating information contained in the underlying geometry and topology of the data. In this article, we build on previous works by introducing a dimensionality reduction method suitable for all real Stiefel manifolds. 

A key feature of Stiefel manifolds is their symmetry with respect to the action of the orthogonal group, the quotient by which is a Grassmannian manifold. In order to produce a representation of the data that respects the symmetry of the underlying space, it is essential to perform dimensionality reduction in an \textit{equivariant} manner. That is, let $\mathbb{F}= \mathbb{R}$ or $\mathbb{C}$, and let $V_k(\mathbb{F}^N)$ be the Stiefel manifold of $k$-planes in $\mathbb{F}^N$. Let $G = \mathcal{O}(k)$ be the orthogonal group that acts via right multiplication on $V_k(\mathbb{F}^N)$. In this work, we describe a technique (see Algorithm \ref{algorithm_stiefel}) that produces an optimal projection $\pi: \mathcal{Y} \mapsto \tilde{\mathcal{Y}} \subset V_k(\mathbb{R}^n)$, for $n<N$ such that $\pi(y \cdot g)= \pi(y) \cdot g$ for $y \in \mathcal{Y}$ and $g \in G$. Because our methodology respects the orthogonal group action on Stiefel manifolds, i.e., is $O(k)$-equivariant, Algorithm \ref{algorithm_stiefel} can also be adapted to handle data living on Grassmannian manifolds, as detailed in Algorithm \ref{algorithm_grassmannian}.

This equivariance property of our projection map is a key feature not generally present in similar dimensionality reduction methods. While \cite{fan2022hyperbolic} produces embeddings of hyperbolic spaces equivariant with respect to the Lorentz group action, this seems to be a relatively recent development. The method of principal nested shape spaces in \cite{dryden2019shapespaces} works with Kendall's shape spaces, which are quotients of spheres by the action of orthogonal groups. However, their approach is not to construct an equivariant projection, but rather to consider only the subset of the so-called pre-shape sphere on which the projection map is a bijection. 

\subsection{Contributions and paper outline}

Here, we summarize our contributions.
\begin{enumerate}
    \item Algorithmically, we establish a dimensionality reduction algorithm, \textit{Principal Stiefel Coordinates (PSC)}, specifically for data living on Stiefel manifolds. Since our algorithm is $O(k)$-equivariant and data on Grassmanian manifolds can be easily lifted to and processed in Stiefel manifolds, our algorithm can be applied to a broader set of data compared to prior work on Grassmanian manifolds.
    \item Theoretically, we define a projection map $\pi_\alpha$ that maps data on a high-dimensional Stiefel manifold onto the image of an embedded low-dimensional Stiefel manifold, and prove its continuity and $O(k)$-equivariance. Furthermore, we show that when data lies \textit{exactly} on a lower-dimensional space, applying an extension of PCA leads to a mapping that \textit{globally} minimizes the projection error despite it being a non-convex optimization problem. However, this may not the best choice when the data does not follow these strict assumptions. In these cases, we demonstrate how an additional gradient descent step results in an improved outcome. 
    \item Experimentally, we demonstrate PSC on both low-dimensional and high-dimensional simulation data, as well as  brain connectivity matrices, videos, simulated neuronal stimulus space model data, and point clouds arising from vector bundles over various manifolds. We compare our results to existing methods such as multidimensional scaling (MDS), persistent cohomology, and principal geodesic analysis (PGA). Our Python code is available publicly on Github.
\end{enumerate}

In Section \ref{sec:prelim}, we specify our notation and provide background information on Stiefel manifolds, singular value decomposition and polar decomposition. In Section \ref{sec:projection}, we define the projection map $\pi_\alpha$ and prove that it is continuous (Theorem \ref{thm:continuous}), equivariant (Proposition \ref{prop:projection_map_equivariant}), and minimizes distortion of the data (Propositions \ref{prop:distanceminimizing} and \ref{prop:distanceminimizing2}). We also give a theoretical guarantee on the broad applicability of our pipeline (Section \ref{sec_alpha_transpose_y_full_rank}). In Section \ref{sec:optimization}, we describe two processes for finding a suitable $\alpha$: one via a PCA-type algorithm adapted to Stiefel manifolds ($\alpha_{PCA}$), and one that further minimizes data fit error using gradient descent ($\alpha_{GD}$). We then prove that $\pi_{\alpha_{PCA}}$ leads to a desirable mapping when the data lies exactly on a lower-dimensional space (Theorem \ref{thm:pca_global_max}). In Section \ref{sec:experiments}, we test our algorithm on various data sets and observe that $\pi_{\alpha_{GD}}$ achieves significantly improved outcomes when the data is nonlinearly generated, as in Sections \ref{sec:stimulus-space-model} and \ref{sec:gdnonlinear}.

\subsection{Main algorithms}\label{subsec_main_pipeline}

To outline the main mathematical contents of the paper, we summarize below the essential steps in our dimensionality reduction pipeline. While we freely use notation that will be defined later in the paper, we hope the following will be a useful reference for the reader.

\bigskip
\bigskip 

\begin{breakablealgorithm}\label{algorithm_stiefel}
\caption{Principal Stiefel Coordinates (PSC): Dimensionality reduction on Stiefel manifold}
\begin{algorithmic}[1]
    \STATE Suppose as given a data set $\mathcal{Y} \subset V_k(\RR^N)$ assumed to live near the image of a linearly-embedded lower-dimensional Stiefel manifold.
    \STATE Fix a user-chosen target dimension $n \ll N$ for the dimensionality reduction.
    \STATE For any $\alpha \in V_n(\RR^N)$, note that $\alpha$ defines an isometric embedding $V_k(\RR^n) \xrightarrow{\alpha} V_k(\RR^N)$; see Figure \ref{fig:pipeline_alpha} and Proposition \ref{prop:isometricembedding}. Let $\pi_\alpha \colon T_\alpha \to \im(\alpha)$ be the continuous, distance-minimizing, equivariant projection that contains a neighborhood of radius $\sqrt{2}$ about $\im(\alpha)$ constructed in Section \ref{sec:projection}.
    \STATE \label{step_warm_start} Construct $\alpha_{PCA} \in V_n(\RR^N)$ as in Definition \ref{def:pca}.

 \STATE Determine if any data points lie outside of $T_{\alpha_{PCA}}$ and remove them from consideration.

    \STATE \label{step_gradient_descent} Using $\alpha_{PCA}$ as the initialization, perform gradient descent on $V_n(\RR^N)$ to find
    \begin{equation}
        \alpha_{GD} \colon = \displaystyle\arg\min_{\alpha \in V_n(\RR^N)}\frac{1}{|\mathcal{Y}|}\sum_{y \in \mathcal{Y}} \|y - \pi_\alpha(y)\|_\mathrm{F}^2.\label{alphaGD}
    \end{equation}
    
 \STATE Determine if any data points lie outside of $T_{\alpha_{GD}}$ and remove them from consideration.

    \STATE \label{step_y_hat} For each $y \in \mathcal{Y}$, consider the point $\pi_{\alpha_{GD}}(y) = \alpha_{GD}(\hat{y}_{\alpha_{GD}})$ as detailed in Definition \ref{defn_projection_map}. Let $\pi_{\alpha_{GD}}(\mathcal{Y})$ be the set $\sett{\pi_{\alpha_{GD}}(y) \mid y \in \mathcal{Y}} \subset V_k(\RR^N)$.
    \STATE \label{step_final_projection} Obtain the representation $\pi_{\alpha_{GD}}(\mathcal{Y}) \subset \alpha_{GD}(V_k(\RR^n))$ of the data $\mathcal{Y}$ as a subset of the isometric image of $V_k(\RR^n)$ under $\alpha_{GD}$.
\end{algorithmic}
\end{breakablealgorithm}

\bigskip

\begin{figure}[t]
    \begin{subfigure}[b]{0.47\linewidth}
       \centering
 \includegraphics[width=\linewidth]{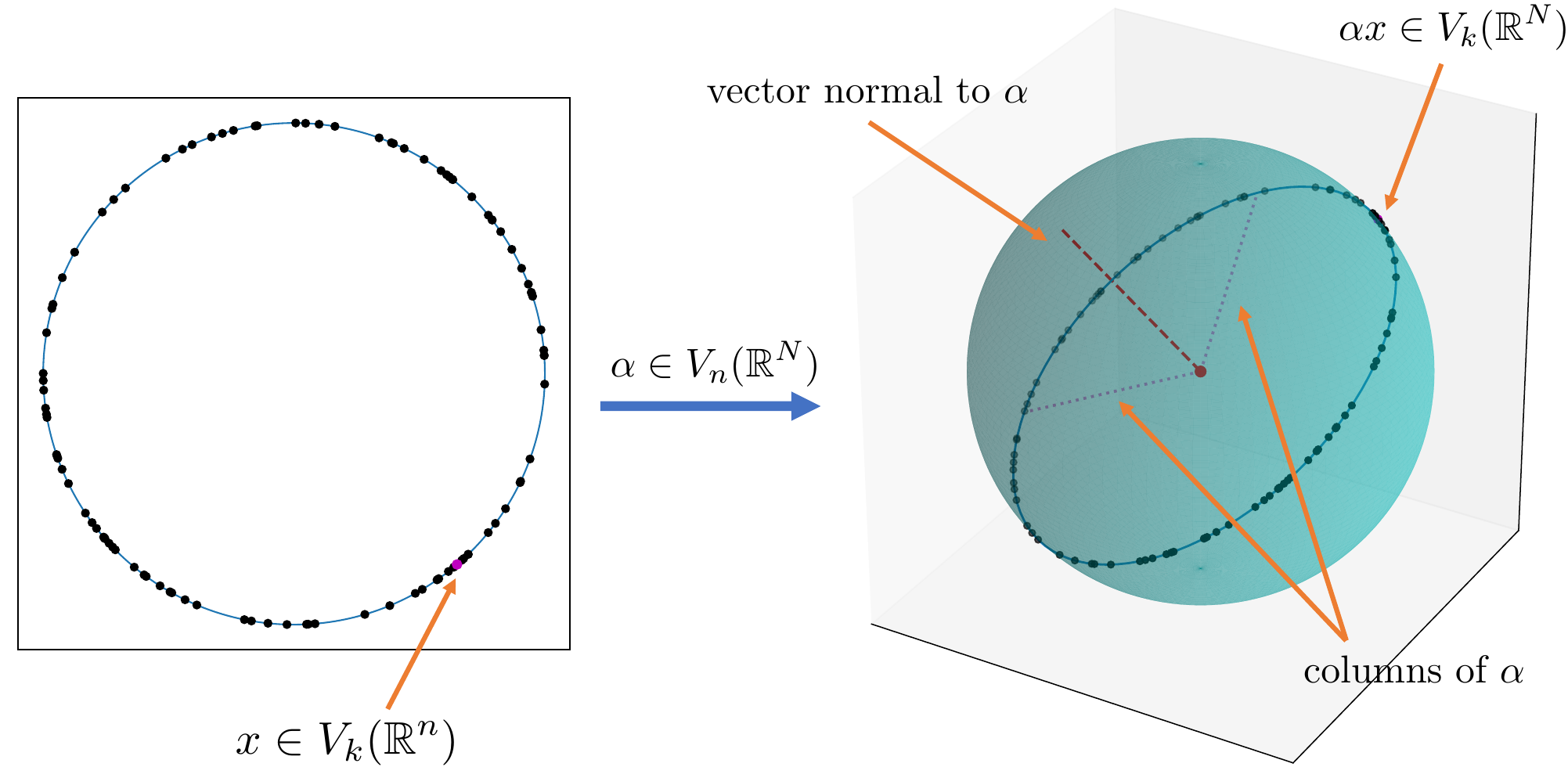}   
    \caption{$\alpha \in V_n(\mathbb{R}^N)$ defines an isometric embedding from $V_k(\mathbb{R}^n)$ to $V_k(\mathbb{R}^N)$. }
    \label{fig:pipeline_alpha}
    \end{subfigure}
\hfill
\begin{subfigure}[b]{0.47\linewidth}
    \includegraphics[width=\linewidth]{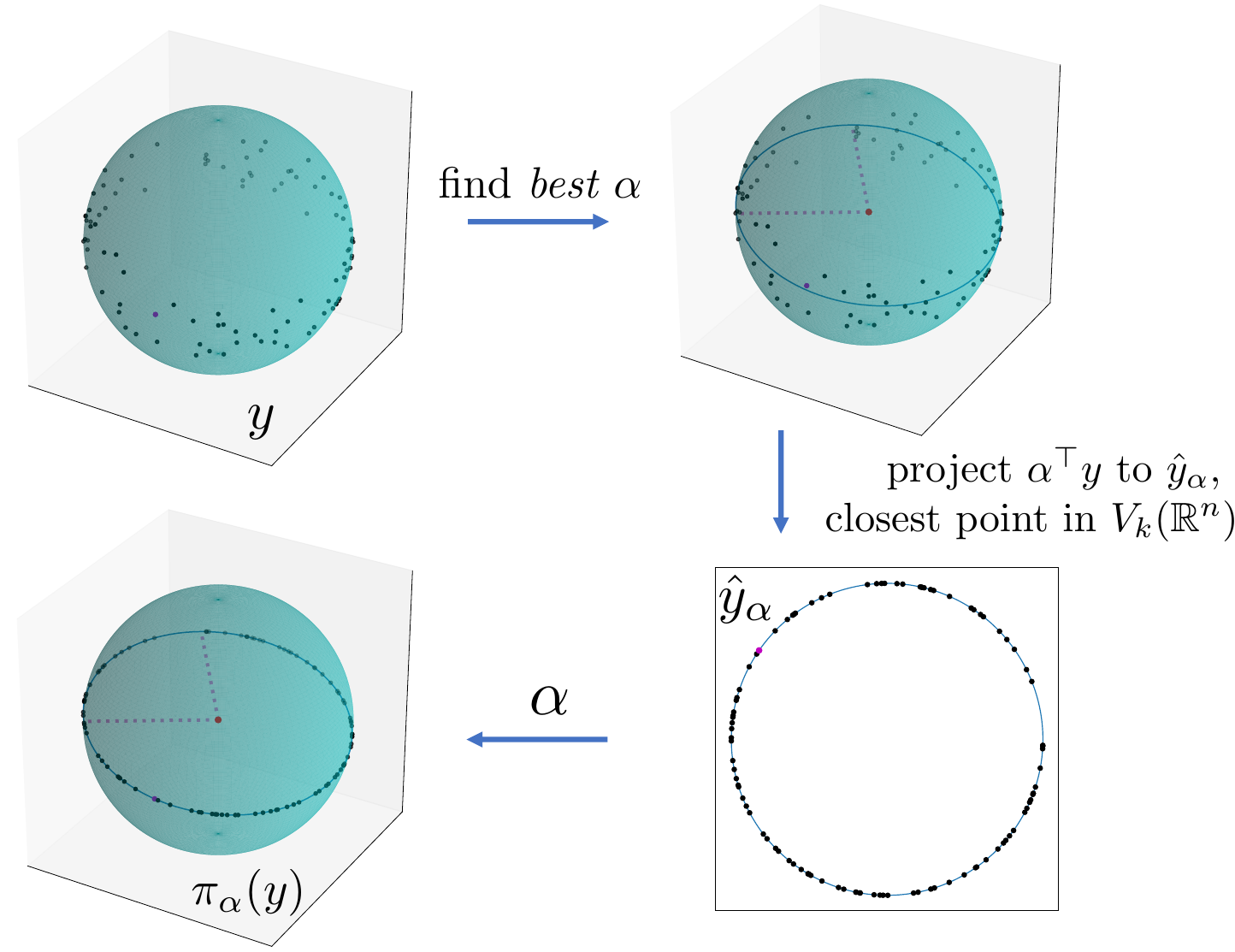}
    \caption{Outline of PSC Algorithm \ref{algorithm_stiefel}. }
    \label{fig:pipeline_overview}
\end{subfigure}
    \caption{Principal Stiefel Coordinates (PSC) illustrated with $k=1, n=2, N=3$. 
    $y, \pi_\alpha(y) \in V_k(\mathbb{R}^N), \alpha \in V_n(\mathbb{R}^N), \hat{y}_\alpha \in V_k(\mathbb{R}^n)$.}
\label{fig:alpha_embedding_and_pipeline_overview}
\end{figure}

To make the assumption in Step 1, above, more precise: we assume that $\mathcal{Y}$ is near $\alpha(V_k(\RR^n))$ for an unknown $\alpha$ and $n$. We offer suggestions on choosing $n$ in Section \ref{appendix_choosing_n} and the PSC algorithm finds an optimal $\alpha$.

\begin{remark}
    Steps 5 and 7 in Algorithm \ref{algorithm_stiefel} may remove data points from consideration. Because we prove in Proposition \ref{prop_tubular_neighborhood} that the domain $T_\alpha$ contains a neighborhood of radius $\sqrt{2}$ about $\im(\alpha)$, if a significant number of data points are removed by Steps 5 and 7, this suggests that the data under consideration does not lie near $\im(\alpha)$ and the PSC algorithm may not be appropriate for the analysis of such data. We also highlight that $T_\alpha$ may contain points outside of the neighborhood of radius $\sqrt{2}$ about $\im(\alpha)$ and we leave it to the user's discretion whether or not to include such points in their analysis.
\end{remark}

It is worth noting that from an implementation perspective, the image of a data point $y \in \mathcal{Y}$ under our dimensionality reduction pipeline is the point $\hat{y}_{\alpha_{GD}} \in V_k(\RR^n)$ as in Step \ref{step_y_hat}. That is, to recover data living on $V_k(\RR^n)$ from Algorithm \ref{algorithm_stiefel}, one would apply the left inverse $\alpha_{GD}^\top$ of $\alpha_{GD}$ to $\pi_{\alpha_{GD}}(\mathcal{Y})$, which simply recovers $\hat{y}_{\alpha_{GD}}$. Indeed, for any $y \in \mathcal{Y}$, we have
\begin{equation}
    \alpha_{GD}^\top \pi_{\alpha_{GD}}(y) = \alpha_{GD}^\top \alpha_{GD}(\hat{y}_{\alpha_{GD}}) = \hat{y}_{\alpha_{GD}}.
\end{equation}
However, from a conceptual standpoint, we find it useful to think of our pipeline as 1) embedding $V_k(\RR^n)$ in $V_k(\RR^N)$ via $\alpha_{GD}$, then 2) performing a projection from $V_k(\RR^N)$ to $\alpha_{GD}(V_k(\RR^n))$, and finally 3) ``pulling back'' to data in $V_k(\RR^n)$. Moreover, the gradient descent in Step \ref{step_gradient_descent} uses the projected points $\pi_\alpha(y) \in V_k(\RR^N)$ in an essential way. The right-hand side of Figure \ref{fig:alpha_embedding_and_pipeline_overview} gives an overview of our pipeline.

Because Algorithm \ref{algorithm_stiefel} uses an equivariant projection $\pi_{\alpha_{GD}}$, our technique can also be applied to data living on Grassmannian manifolds $Gr(k, \RR^N)$. In further detail, for any $y \in V_k(\RR^N)$ and $g \in O(k)$, it follows from Proposition \ref{prop:projection_map_equivariant} that $\pi_{\alpha_{GD}}(y)g = \pi_{\alpha_{GD}}(yg)$. Hence, we have that
\begin{equation}
    \alpha_{GD}^\top\pi_{\alpha_{GD}}(y)g = \alpha_{GD}^\top\pi_{\alpha_{GD}}(yg).
\end{equation}

This shows that the algorithm below is well-defined.\\

\begin{breakablealgorithm}\label{algorithm_grassmannian}
\caption{Dimensionality reduction for Grassmannian manifolds}
    \begin{algorithmic}[1]
    \STATE Suppose as given a data set $\bar{\mathcal{Y}} \subset Gr(k, \RR^N)$ and assume that $\bar{\mathcal{Y}}$ lives near the image of a linearly-embedded lower-dimensional Grassmannian manifold.
    \STATE For each $\bar{y} \in \bar{\mathcal{Y}}$, choose an orthonormal basis of $\bar{y}$ to obtain an associated point $y \in V_k(\RR^n)$, and let $\mathcal{Y} = \sett{y \mid \bar{y} \in \bar{\mathcal{Y}}}$. Note that any two orthonormal bases $y$ and $y'$ are related by an element of $O(k)$. That is, $y' = yg$ for some $g \in O(k)$.
    \STATE Perform Algorithm \ref{algorithm_stiefel} on the set $\mathcal{Y}$ to obtain $\mathcal{X} \subset V_k(\RR^n)$.
    \STATE Project $\mathcal{X} \subset V_k(\RR^n)$ to $Gr(k, \RR^n)$ via the canonical quotient map $V_k(\RR^n) \to Gr(k, \RR^n)$ to obtain the set $\bar{\mathcal{X}} \subset Gr(k, \RR^n)$.
    \STATE Recover a representation $\bar{\mathcal{X}}$ of the original data as a subset of $Gr(k, \RR^n)$. 
    \end{algorithmic}
\end{breakablealgorithm}

\medskip

\section{Preliminaries}\label{sec:prelim}

This section is devoted to fixing notation and summarizing standard results that will be used throughout the rest of the paper. For a handy reference, Table \ref{tab:notations} details our most commonly used notation. 

\subsection{Notations}

The notation $\RR^{s \times t}$ denotes an $s \times t$ matrix with entries in $\mathbb{R}$, $I_s$ an $s \times s$ identity matrix, and $0_{s \times t}$ an all-zero matrix of size $s \times t$. We use $A^\top$ and $A^*$ to denote the transpose and hermitian of the matrix $A$, respectively. The trace of $A$ is $\tr(A)$, $\|A\|_\mathrm{F} $ is the Frobenius norm, $\|A\|_{*} $ is the nuclear norm, and $\|A\|_{op}$ is the operator norm, i.e. the largest singular value. $|S|$ is the cardinality of the set $S$, and $\im(A)$ is the image of $A$. 

\begin{table}[ht]
\footnotesize
    \centering
    \begin{tabular}{l p{11cm}}
    \toprule
        \textbf{Notation} & \textbf{Meaning} \\ \midrule
       $k,~n,~N$  & Data dimensions $0< k \leq n \ll N$. \\
      $V_{t}(\RR^s)$& Stiefel manifold of orthonormal $t$-frames in $\RR^s$ (c.f. \ref{defn:stiefel}). \\
      $Gr(t, \RR^s)$& Grassmanian manifold of $t$-dimensional subspaces in $\RR^s$. \\      
      $O(t)$ & Orthogonal group in $\RR^t $. \\
      $x,~ \mathcal{X}$ & Data points on $V_{k}(\RR^n)$, $x \in \mathcal{X}$.\\
      $y,~ \mathcal{Y}$ & Data points on $V_{k}(\RR^N)$, $y \in \mathcal{Y}$.\\
      $\alpha$ & A point in $V_{n}(\RR^N)$ (c.f. \ref{prop:isometricembedding}).\\
    $\alpha_{PCA}$ & The warm-start embedding used in Step \ref{step_warm_start} of Algorithm \ref{algorithm_stiefel} (c.f. \ref{def:pca}).\\
    $\alpha_{GD}$ & The embedding produced through gradient descent in Step \ref{step_gradient_descent} of Algorithm \ref{algorithm_stiefel}. \\      
      $\pi_\alpha$ & A projection map from a suitable subset of $V_{k}(\RR^N)$ to $\alpha(V_{k}(\RR^n))$ (c.f. \ref{defn_projection_map}).\\
      $\hat{y}_\alpha$ & A point on $V_{k}(\RR^n)$ that is closest to $\alpha^\top y$ (c.f. \ref{defn_projection_map}).\\
      \bottomrule
    \end{tabular}
    \caption{Frequently referenced notations.}
    \label{tab:notations}
\end{table}

\subsection{Stiefel manifolds}

For these definitions, we use arbitrary dimensions $0 < t < s$. 
\begin{definition}[Stiefel manifold] \label{defn:stiefel}

An $s\times t$ matrix $A$ is an element of the \textbf{Stiefel manifold} $V_t(\RR^s)$ if and only if $A^\top A=I_t$.
\end{definition}

The $V_t(\RR^s)$ Stiefel manifolds are by definition embedded in the Euclidean space of matrices $\RR^{s\times t}$ and inherit the following distance. 

\begin{definition}\label{def:distance}
For $P,Q\in \RR^{s \times t}$, define the \textbf{Frobenius inner product} to be
\begin{equation}
    \langle P,Q\rangle_{F} = \tr(P^\top Q). \label{IP1}
\end{equation}
and let $D(P,Q) = \|P-Q\|_\mathrm{F}$ be the induced distance.

Also let the \textbf{nuclear norm} of a matrix $A$ be the sum of all singular values of $A$, i.e. $\|A\|_*=\tr(\Sigma_r)$,

where $\Sigma_r$ is defined in \ref{defn:svd}.

\end{definition}

Any Stiefel manifold $V_t(\RR^s)$ has a left (resp. right) action of $O(s)$ (resp. $O(t)$) given by matrix multiplication, and this action is an isometry with respect to $D$.

\subsection{Matrix decompositions}
Our method makes frequent use of the singular value and polar decomposition of matrices, recorded below for the sake of completeness.

\begin{definition}[Singular value decomposition (SVD)] \label{defn:svd}
For any matrix $A \in \RR^{s\times t}$, there exists a set of matrices such that
\begin{equation}
    A = P \begin{bmatrix}
\Sigma_r & 0\\
0 & 0_{s-r, t-r}
\end{bmatrix} Q^\top
\end{equation}
where $r \le \min\{s, t\}$ is the rank of $A$, $P \in \RR^{s \times s}, Q \in \RR^{t \times t}, P^\top P = I_s, Q^\top Q = I_t$, and $\Sigma_r \in \RR^{r \times r}$ is a diagonal matrix with nonnegative real numbers on the diagonal. 

\end{definition}
We note that while one matrix may have many forms of SVD, the set of singular values for a matrix $A$ is unique. Related to the singular value decomposition of a matrix is its polar decomposition \cite[Theorem 8.1, p. 193]{Higham_functions_of_matrices}.

\begin{theorem}\label{thm_form_of_polar} (Polar decomposition)
Let $A \in \RR^{s\times t}$ with $s \geq t$ be full rank. There exists a matrix $U \in \RR^{s \times t}$ with orthonormal columns and a unique self-adjoint positive semidefinite matrix $H \in \RR^{t \times t}$ such that $A = UH$. The matrix $H$ is given by $H = (A^\top A)^{1/2}$. 

\end{theorem}

\section{Projection \texorpdfstring{$\pi_{\alpha}$}{pi\_a}}\label{sec:projection}

This section defines and proves several key properties of the projection $\pi_\alpha$ referenced in Step 3 of Algorithm \ref{algorithm_stiefel}. Its main result is Theorem \ref{thm_main_thm_equivariant_projection}, which uses the following Definition.

\begin{definition}\label{def:isometric_embedding}
    Let $\alpha \in V_n(\RR^N)$. By slight abuse of notation, we let $\alpha \colon V_k(\RR^n) \to V_k(\RR^N)$ be the map given by left matrix multiplication by $\alpha$.
\end{definition}

Proposition \ref{prop:isometricembedding} shows that this map is, in fact, an isometric embedding. We now define the projection associated to this map.

\begin{definition}\label{defn_projection_map} (Projection.)
Fix $\alpha \in V_n(\RR^N)$. Let $L_\alpha= \{y \in V_k(\RR^N) ~~|~~ \rank(\alpha^\top y)<k \}$. Define \[ \pi_\alpha \colon V_k(\RR^N) \setminus L_\alpha \to \alpha(V_k(\RR^n)) \] as follows for an element $y$ in the domain. Let $\hat{y}_\alpha$ be the polar factor $U$ in the polar decomposition of $\alpha^\top y$. Then $\hat{y}_\alpha$ has orthonormal columns so $\hat{y}_\alpha \in V_k(\RR^n)$. Define $\pi_\alpha$ by $\pi_\alpha(y) := \alpha\hat{y}_\alpha.$
\end{definition}

\begin{remark}\label{rem_equiv_def_of_projection}
Note that on the set $V_k(\RR^N) \setminus L_\alpha$, we can define $\pi_\alpha$ equivalently as follows. For $y \in V_k(\RR^N) \setminus L_\alpha$, let $\alpha^\top y = P \begin{bmatrix}
\Sigma_k \\
0 
\end{bmatrix}Q^\top $ be any singular value decomposition of $\alpha^\top y$. Let $\hat{y}_\alpha$ be given by
\begin{equation}
    \hat{y}_\alpha = P \begin{bmatrix}
    I_k\\
    0
    \end{bmatrix}Q^\top 
\end{equation}  
and let $\pi_\alpha(y) = \alpha(\hat{y}_\alpha).$ By the uniqueness of the polar decomposition in the full rank case, $\hat{y}_\alpha$ is independent of the particular choice of singular value decomposition of $\alpha^\top y$.
\end{remark}

We now state the main result of this section, which is that the projection map has the desired properties.

\begin{theorem}\label{thm_main_thm_equivariant_projection}
    Fix an integer $n$ and element $\alpha \in V_n(\RR^N)$ and consider the associated isometric embedding $\alpha \colon V_k(\RR^n) \to V_k(\RR^N)$. Let $T_\alpha \subset V_k(\mathbb{R}^N)$ be the set of points $y \in V_k(\RR^N)$ such that $\rank(\alpha^\top y) = k$. Then $T_\alpha$ contains all points within $\sqrt{2}$ of $\im(\alpha)$ and there is a continuous, $O(k)$-equivariant projection $\pi_\alpha \colon  T_\alpha \to \alpha(V_k(\RR^n)) = \im(\alpha)$ such that for each fixed $y \in T_\alpha$, the following holds: 
\begin{equation}\label{eqn_distance_minimizing}
    D(y, \pi_\alpha(y)) \leq D(y,x) \text{ for all } x \in \alpha(V_k(\RR^n)).
\end{equation}
\end{theorem}
\begin{proof}
The projection $\pi_\alpha$ is defined in Definition \ref{defn_projection_map}. Proposition \ref{prop_tubular_neighborhood} shows that the domain of $\pi_\alpha$ (as in Definition \ref{defn_projection_map}) includes the radius $\sqrt{2}$-neighborhood described above. The equivariance of $\pi_\alpha$ is Proposition \ref{prop:projection_map_equivariant}. The claim in \eqref{eqn_distance_minimizing} follows from Propositions \ref{prop:distanceminimizing} and \ref{prop:distanceminimizing2}. Continuity of the projection is Theorem \ref{thm:continuous}.
\end{proof}

Using Definition \ref{defn:stiefel} and invariance of trace under cyclic permutations, a straightforward calculation establishes the following.

\begin{proposition}\label{prop:isometricembedding} 
For each $\alpha \in V_n(\RR^N)$, the associated map from Definition \ref{def:isometric_embedding} is an isometric embedding (with respect to $D$) of $V_k(\RR^n)$ into $ V_k(\RR^N)$.
\end{proposition}

Figure \ref{fig:pipeline_alpha} shows a low-dimensional example of Proposition \ref{prop:isometricembedding}.
Based on this observation, we propose the following approach. We seek an ``optimal" embedding $\alpha$ that describes our data, then use that $\alpha$ to define an $O(k)$-equivariant projection $\pi_\alpha: \mathcal{Y} \mapsto \alpha(\mathcal{Y}) \subset \alpha(V_k(\mathbb{R}^n))$. Figure \ref{fig:pipeline_overview} outlines our algorithm.
We now define such a projection $\pi_\alpha$ for any given $\alpha$. In Section \ref{sec:optimization}, we discuss the problem of optimizing $\alpha$ over $V_n(\RR^N)$ given the dataset. 

The remainder of this section is devoted to proving that $\pi_\alpha$ is well-defined on a suitable neighborhood in $V_k(\RR^N)$ (Proposition \ref{prop_tubular_neighborhood}), that $\pi_\alpha$ is continuous on this domain (Proposition \ref{thm:continuous}), that $\pi_\alpha$ is equivariant (Proposition \ref{prop:projection_map_equivariant}) and that \eqref{eqn_distance_minimizing} holds (Propositions \ref{prop:distanceminimizing} and \ref{prop:distanceminimizing2}).

\begin{proposition}\label{prop_tubular_neighborhood}
Given $\alpha\in V_n(\RR^N)$, $x\in V_k(\RR^n)$, $y\in V_k(\RR^N)$, let $\im(\alpha)^\perp := \{y\in V_k(\RR^N)\,|\, \langle y_i, \alpha_j\rangle = 0\,\forall i,j\}$ so the columns $y_i$ of $y$ are pairwise orthogonal to the columns of $\alpha$ as vectors in $\RR^N$. We have that:
\begin{enumerate}[label=(\alph*)]
\item If $y\in \im(\alpha)^\perp$ then $\|y-\alpha x\|_\mathrm{F}=\sqrt{2k}$. 
\item If $\|y-\alpha x\|_\mathrm{F}<\sqrt{2}$ then $\alpha^\top y$ is full rank, and so the map $\pi_\alpha$ is well-defined. 
\end{enumerate}
Thus the map $\pi_\alpha$ is well-defined on all points at distance less than $\sqrt{2}$ from some element of $\im(\alpha)$.
\end{proposition}

\begin{proof}
First, note that for any $x \in V_t(\RR^s)$, we have $\|x\|_\mathrm{F}^2 = \|xx^\top\|_\mathrm{F}^2 = t$.

For (a), suppose that    $y\in \im(\alpha)^\perp$. Then, for $ x \in V_k(\RR^n)$,  
\begin{equation}
     D(y,\alpha x)=\|y - \alpha x \|_\mathrm{F}=\sqrt{\|y\|_\mathrm{F}^2 - 2 \langle y, \alpha x \rangle +\|\alpha x \|_\mathrm{F}^2} = \sqrt{k-2\cdot 0+k}.
\end{equation}

For (b), we prove the contrapositive: assuming that $\alpha^\top y$ is not full rank, we show that the distance is bounded below by $\sqrt{2}$. Let $\alpha_i$ denote the $i$th column of $\alpha$, and let $y_i$ denote the $i$th column of $y$. If $\alpha^\top y$ does not have full rank, then there must be some dependence relation among the column vectors $\sett{\alpha^\top y_1, \alpha^\top y_2, \ldots, \alpha^\top y_k}$, meaning that $\text{span}(\sett{y_i})$ intersects $\ker(\alpha^\top) = \text{span}(\sett{\alpha_i})^\perp$ non-trivially, i.e.,  $\text{span}(\{y_i\})\cap \text{span}(\{\alpha_i\})^\perp$ has nonzero dimension $l>0$. Up to right multiplication by an element of $O(k)$, we may assume the first $l$ columns of $y$ lie entirely in the orthogonal complement of \text{span}(\{$\alpha_i\})$.

In these coordinates, we may compute the distance between $y$ and any element $a\in \alpha(V_k(\mathbb{R}^n))$ as the sum of the squared norms of the columns of $y-a$. However, for $1\leq i \leq l$, we know that $y_i$ and $a_i$ are unit vectors that are orthogonal to one  another, and hence $||y_i-a_i|| = \sqrt{2}$. Thus, $D(y, a) = ||y - a||_\mathrm{F} \geq \sqrt{2}l \geq \sqrt{2}$.
\end{proof}

\begin{proposition}\label{thm:continuous}
    Fix $\alpha \in V_n(\mathbb{R}^N)$ and let 
    \begin{equation}
        \pi_\alpha \colon V_k(\mathbb{R}^N)\setminus L_\alpha \to \alpha(V_k(\mathbb{R}^n))
    \end{equation}
    be as in Definition \ref{defn_projection_map}. On this (restricted) domain, $\pi_\alpha$ is continuous.
\end{proposition}
\begin{proof}
 
    Denote by $f$ the assignment of a matrix $A$ to the positive semidefinite matrix $H$ in the polar decomposition of $A$. It follows from, for instance, \cite[Theorem 8.9]{Higham_functions_of_matrices} that $f$ is continuous. Therefore, $\pi_\alpha$ is given by $y \mapsto \alpha \alpha^\top y ( f(\alpha^\top y)^{-1})$
    which is a composition of continuous functions. 
\end{proof}

\begin{proposition}\label{prop:projection_map_equivariant}
The map $\pi_\alpha$ of Definition \ref{defn_projection_map} is $O(k)$-equivariant. That is, for any $g \in O(k)$, we have $\pi_\alpha(yg) = \pi_\alpha(y)g$. Note that the group action is right matrix multiplication in both cases.
\end{proposition}

\begin{proof}

The proof follows from the singular value decomposition of $\alpha^\top y$ and $(\alpha^\top y)g$:
\begin{equation}
\alpha^\top y = P \begin{bmatrix}
\Sigma_k \\
0 
\end{bmatrix}Q^\top, \quad \quad \alpha^\top yg = P \begin{bmatrix}
\Sigma_k \\
0 
\end{bmatrix}(g^\top Q)^\top.  
\end{equation}
\end{proof}

The following two propositions together show that our projection minimizes the distance each $y$ travels to $\pi_\alpha(y)$ for a fixed $\alpha$, and thus minimizes ``distortion'' of the data.
\begin{proposition}\label{prop:distanceminimizing}
For a fixed $\alpha \in V_n(\RR^N)$ and distance D, let $y \in V_k(\RR^N)$ $\setminus L_\alpha$. As defined above, $\hat{y}_\alpha \in V_k(\RR^n)$ achieves the minimum distance between $\alpha^\top y \in \RR^{n \times k}$ and points in $V_k(\RR^n)$. 
\end{proposition}
\begin{proof}
It is shown in \cite[8.4]{Higham_functions_of_matrices}, where it is attributed to \cite{Fan_Hoffman_matrices}, that
\begin{equation}
    \hat{y}_\alpha = \arg\min_{Q \in \CC^{n\times k}}\sett{\norm{\alpha^\top y - Q} \mid Q^*Q = I_k}
\end{equation}
for any unitarily invariant norm. Moreover, when $\alpha^\top(y)$ has full rank, $\hat{y}_\alpha$ is unique and the argmin is a single point. Since
\begin{equation}
\min_{Q \in \CC^{n\times k}}\sett{\norm{\alpha^\top y - Q} \mid Q^*Q = I_k} \leq \min_{Q \in \RR^{n\times k}}\sett{\norm{\alpha^\top y - Q} \mid Q^*Q = I_k}
\end{equation}
and $\hat{y}_\alpha$ in fact has real entries, the fact that $D$ is unitarily invariant when viewed as a norm on $\CC^{n \times k}$ completes the argument.
A direct proof can be found in \cite{Kahan11}.
\end{proof}

\begin{proposition}\label{prop:distanceminimizing2}
Fix $\alpha \in V_n(\RR^N)$, and $y \in V_k(\RR^N)$. Let $x \in V_k(\RR^n)$ be a (not necessarily unique) point of minimal distance on  $V_k(\RR^n)$ to $\alpha^\top y \in \RR^{n\times k}$. 
Then a (not necessarily unique)  point of minimal distance on $\alpha(V_k(\RR^n))$ to $y$ is $\alpha x$.
\end{proposition}
\begin{proof}
The proof follows from the following calculation:
\begin{align*}
x = \arg\min_{u \in V_k(\RR^n)}~ \| \alpha^\top y - u\|_\mathrm{F} &= \arg\min_{u \in V_k(\RR^n)}~ \| \alpha^\top y\|_\mathrm{F}^2-2\langle \alpha^\top y, u \rangle + k \\
&= \arg\min_{u \in V_k(\RR^n)}~D(y, \alpha u).
\end{align*}
\end{proof}

\subsection{When \texorpdfstring{$n<<N$}{n<<N}, \texorpdfstring{$\alpha^\top y$}{a'y} is generically full rank}\label{sec_alpha_transpose_y_full_rank}

The purpose of this section is to prove Corollary \ref{codimcount}, which provides a guarantee on when the rank condition $\rank(\alpha^\top y) = k$ will be satisfied. Let $M(n,k)$ be the $\mathbb{R}$-vector space of $n\times k$ matrices. First, we recall some standard facts about tangent spaces to Stiefel manifolds (see, e.g., \cite{boumal2023intromanifolds}).

\begin{proposition}
The tangent space to the Stiefel manifold $V_k(\mathbb{R}^N)$ at a point $y$, denoted $T_y(V_k(\mathbb{R}^N))$, can be identified with matrices of the form
\begin{equation}
yC + y_\perp D \label{stiefeltangent}
\end{equation}
where $C$ is a $k\times k$ skew-symmetric matrix, $y_\perp$ is an $(N-k)\times k$ matrix such that $y_\perp y_\perp^\top$ is projection to the orthogonal complement to span$(y)$, and $D$ is any $(n-k)\times k$ matrix.
\end{proposition}

\noindent It will be useful to consider variations in both $\alpha$ and $y$, so we define the map 

\begin{eqnarray}\label{mudefn}
\mu: V_n(\mathbb{R}^N) \times V_k(\mathbb{R}^n) & \to & M(n, k)\\
(\alpha, y) & \mapsto & \alpha^\top y.\nonumber
\end{eqnarray}

\noindent The determinant $\det(\mu(\alpha, y)^\top\cdot \mu(\alpha, y))$ is nonzero if and only if $\alpha^\top y$ satisfies the rank condition, so the set of such $(\alpha, y)$ is an open submanifold of the product $V_n(\mathbb{R}^N)\times V_k(\mathbb{R}^N)$. The same is true when restricting $\mu$ to specific $\alpha, y$. By uniform continuity of the determinant, given $\mathcal{Y}$ any compact set of points such that $\alpha_0^\top y$ is full rank for all $y\in\mathcal{Y}$, there is a neighborhood $U$ about $\alpha_0$ such that $\alpha^\top y$ is full rank for all $y\in \mathcal{Y}$ and $\alpha\in U$. However, we will show a more global result.

The complement of the set of full-rank $(\alpha, y)$ consists of those $y$ such that $\alpha^\top y$ has rank strictly less than $k$. We will directly compute the (positive) codimension of this set. For instance, if $\alpha^\top y$ has rank 0, this only occurs when $\mbox{span}(y)$ is orthogonal to $\mbox{span}(\alpha)$. The set of such $y$ is a Stiefel manifold $V_k(\mathbb(\mbox{Ran}(\alpha)^\perp)$, which certainly has dimension less than that of $V_k(\mathbb{R}^N)$.

For ranks $1,\dots, k-1$, we appeal to the following well-studied space, see e.g. \cite{harris1992algebraic} under the name of determinantal varieties, or alternately as the manifold of fixed-rank matrices
\begin{equation*}
M(r, M, N) := \{ X \in \mathbb{R}^{M\times N}\,|\, \rank(X) = r\}.
\end{equation*}

\noindent The following fact is well-known but appears as formulated below in \cite[Theorem 3.1]{schneideruschmajew}.

\begin{proposition}\label{fixedranktangent}
$M(r, M, N)$ is a smooth manifold of dimension $r(M+N-r)$. For $X\in M(r, M, k)\subset \mathbb{R}^{m\times k}$, let $U = \mbox{Ran}(X)$, and $V = \mbox{Ran}(X^\top)$. Then 
\begin{equation}
T_X(M(r, M, k)) = (U\otimes V) \oplus (U\otimes V^\perp) \oplus (U^\perp\otimes V) \subset \mathbb{R}^{m\times k}.\label{fixedranktangentdecomp}
\end{equation}
Similarly, the normal space can be identified with $U^\perp\otimes V^\perp$.
\end{proposition}

\noindent Using the inner product to identify $\mathbb{R}^N$ and $\mathbb{R}^k$ with their duals, $U^\perp\otimes V^\perp$ can be thought of as $Hom(V^\perp, U^\perp)$.

The utility of the above decomposition of the tangent space arises from applying the following standard theorem of differential topology; the statement here is from \cite[Theorem 6.30]{leesmoothmanifolds}. Recall that a map $F:N\to M$ is \emph{transverse} to a submanifold $S\subset M$ if for every $n\in N$ with $f(n)\in S$, $dF(T_n(N))+T_{f(n)}(S) = T_{f(n)}(M)$.

\begin{theorem}
Suppose $N$ and $M$ are smooth manifolds and $S\subset M$ is an embedded submanifold. If $F: N\to M$ is a smooth map that is transverse to $S$ , then $F^{-1}(S)$ is an embedded submanifold of $N$ whose codimension is equal to the codimension of $S$ in $M$.
\end{theorem}

Thus, if $\mu$ is transverse to $M(r,n, k)\subset \mathbb{R}^{n\times k}$, then the $\mu^{-1}(M(r,n,k))$ are smooth manifolds of positive codimension partially stratifying the complement of the full-rank $(\alpha, y)$, capped off by the $V_k(\mathbb(\mbox{Ran}(\alpha)^\perp)$ mentioned above. As (\ref{fixedranktangentdecomp}) characterizes the normal space to $M(r,n,k)$, we may verify transversality directly. 
\begin{proposition}
For any $y\in V_k(\mathbb{R}^N)$, the map $\mu_y: V_n(\mathbb{R}^N)\to M(n,k)$ is transverse to $M(r,n,k)$.
\end{proposition}

\begin{proof}
Since $\mu$ in (\ref{mudefn}) is the restriction of a linear map on $M(N,k)$, any element of the image of $d\mu$ has the following form, where $\alpha_\perp$ is any $(N-n)\times n$ matrix such that $\alpha_\perp \alpha_\perp^\top$ is projection to the orthogonal complement of span$(\alpha)$.
\begin{equation}
    (\alpha\cdot C)^\top y + (\alpha_\perp D)^\top y = C^\top\alpha^\top y + D^\top \alpha_\perp^\top y\label{dmualphamatrices}
\end{equation}

using the characterization of $T_y(V_k(\mathbb{R}^N))$ from (\ref{stiefeltangent}). Note that we may take any $\alpha_\perp$ whose columns span the orthogonal complement to $\alpha$. We also use the notation of (\ref{fixedranktangentdecomp}) for $U=\mbox{Ran}(\alpha^\top y)$ and $V=\mbox{Ran}(\alpha^\top y)^\top = \mbox{Ran}(y^\top\alpha)$.

By Proposition \ref{fixedranktangent}, it suffices to show that the projection of $d\mu(T_\alpha(V_n(\mathbb{R}^N))$ to $U^\perp\otimes V^\perp$ is surjective. Interpreting these matrices as elements of $\mbox{Hom}(\mathbb{R}^k, \mathbb{R}^n)$, this is equivalent to showing that, given any $v'\in V^\perp$, matrices of the form (\ref{dmualphamatrices}) can be chosen such that the image of $v'$ is in $U^\perp$. We focus on the second term, and let $u'\in U^\perp$ be arbitrary.

First, $V^\perp$ consists of the orthogonal complement to $\mbox{Ran}(y^\top\alpha)$. By assumption $V = \mbox{Ran}(\alpha^\top y)$ is rank $r<k$, so $\mbox{Ran}(y^\top \alpha)$ is a rank-$r$ subspace of a $k$-dimensional space, the projection of $\mbox{Ran}(\alpha)$ to $\mbox{Ran}(y)$. Thus $V^\perp$ is rank $k-r$. Note that $V^\perp$ is orthogonal to $\mbox{Ran}(\alpha)$.

On the other hand, $U^\perp$ consists of the orthogonal complement to $\alpha^\top y$ in $\mathbb{R}^n$, which has dimension $n-r$. 

Since right multiplication by $O(k)$ does not affect $\mbox{Ran}(\alpha^\top_\perp y)$, we may assume without loss of generality that the first column of $y$ is $v'\in V^\perp\subset \mbox{span}(y)\cap \mbox{span}(\alpha)^\perp$. To prove the claim, it suffices to send the first column of $y$ to $u'$. Moreover, as we may choose any columns for $\alpha_\perp$ which span the complement of $\mbox{Ran}(\alpha)$, we may take the first $k-r$ columns to be the same as those of $y$. Thus the first column of $\alpha^\top_\perp y$ has 1 only in the first entry, with the rest 0. Then for $D^\top$ any $n$ by $N-n$ matrix whose first column is $u'$, $D^\top \alpha^\top_\perp y$ has the desired property.
\end{proof}

\begin{corollary}\label{codimcount}
The inverse images $\mu_\alpha^{-1}(M(r,n,k))$ are open submanifolds of $V_n(\mathbb{R}^N)$, with codimension $nk-r(n+k-r)$ (i.e. dimension $Nn-\frac{1}{2}n(n+1)-nk+r(n+k-r)$).     
\end{corollary}

For a fixed $y$, let $\mathcal{M}_y\subset V_n(\mathbb{R}^N)$ be the set of those $\alpha$ for which $\alpha^\top y$ is full rank, and let $\mathcal{X}_y$ be its complement. As a consequence of Corollary \ref{codimcount}, for a fixed data point $y$, the set of points where $\rank(\alpha^\top y)<k$ has strictly positive codimension. Moreover, even if $y$ varies in a full-dimensional neighborhood $U$ of $V_k(\mathbb{R}^N)$, so long as $n<<N$, the union $\cup_{y\in U} \mathcal{X}_y$ swept out by $\mathcal{Y}$ still has positive codimension in $\dim(V_n(\mathbb{R}^N))$. Hence the set of full-rank $(\alpha, y)$ is of full measure, and a generic choice of $(\alpha, y)$ will satisfy the rank condition.

\section{Optimization for \texorpdfstring{$\alpha$}{a}}\label{sec:optimization}

In this section, we discuss our process for optimizing $\alpha$ over $V_n(\mathbb{R}^N)$, detailing Steps 4 and 5 of Algorithm \ref{algorithm_stiefel}.

\begin{definition}[PCA and $\alpha_{PCA}$]\label{def:pca}
Given zero-mean data points $y\in \mathcal{Y}  \subset V_k(\RR^N)$, consider the matrix $\hat{\Sigma}= \frac{1}{|\mathcal{Y}|}\sum_{y \in \mathcal{Y}} y y^\top \in \mathbb{R}^{N \times N}$. Apply eigendecomposition on $\hat{\Sigma}$ to get $N$ eigenvectors, each of which is an $N$-dimensional vector. Note that since $\hat{\Sigma}$ is a real symmetric matrix, its eigenvalues are real and its eigenvectors can be selected to be orthonormal. Select $n$ of the eigenvectors that correspond to the $n$ highest eigenvalues of $\hat{\Sigma}$, and set them to be columns of $\alpha_{PCA} \in V_n(\RR^N) \subset \mathbb{R}^{N \times n}$. We refer to $\alpha_{PCA}$ as the output of PCA on $\mathcal{Y}$, or $\text{PCA}(\mathcal{Y})$.
\end{definition} 

When $n$ is not predetermined, studying the resulting eigenvalues or singular values and identifying a ``natural gap'' can be a good strategy. See Section \ref{appendix_choosing_n} for an example.

\begin{remark}[Relationship with standard PCA]
    We comment on the relationship between $\alpha_{PCA}$ and standard PCA, using the notation of Algorithm \ref{algorithm_stiefel}. First, let us consider the case $k=1$, noting that $V_1(\RR^t) \cong S^{t-1}$ is a sphere for all $t \geq 1$. In this case, the projection $\pi_{\alpha_{PCA}}$ is equivalent to performing standard PCA (e.g. \cite[Section 12.2]{isl}) on the data $\mathcal{Y}$, followed by normalization of the data onto a sphere. In the case $k > 1$, the embedding $\alpha_{PCA}$ represents a novel extension of standard PCA to elements of Stiefel manifolds. In particular, it applies to data sets where the individual data points are matrices with multiple columns rather than merely column vectors.
\end{remark} 

We conclude our introduction of $\alpha_{PCA}$ with the following remark on equivariance versus invariance.

\begin{remark}[PCA is $O(k)$-invariant]
PCA in Definition \ref{def:pca} is $O(k)$-\textit{in}variant, not $O(k)$-\textit{equi}variant. In other words, for some $g \in O(k)$, the output of PCA on data points $y \in \mathcal{Y}$ is identical to that of PCA on $\{yg~|~y \in \mathcal{Y}\}$. This follows from noticing that
\begin{equation}
    \hat{\Sigma}= \frac{1}{|\mathcal{Y}|}\sum_{y \in \mathcal{Y}} (yg) (yg)^\top = \frac{1}{|\mathcal{Y}|}\sum_{y \in \mathcal{Y}}y(gg^\top) y^\top= \frac{1}{|\mathcal{Y}|}\sum_{y \in \mathcal{Y}} y y^\top.
\end{equation}

However, if we use PCA as a way to choose an $\alpha$ to use in $\pi_{\alpha}$ in Definition \ref{defn_projection_map}, then our entire projection algorithm (which includes both finding $\alpha$ from $y$ and applying $\pi_\alpha$ to $y$) is $O(k)$-\textit{equi}variant. That is, 
\begin{equation}
    \pi_{\mathrm{PCA}(yg)}(yg) = \pi_{\mathrm{PCA}(y)}(yg)=\pi_{\mathrm{PCA}(y)}(y)g.
\end{equation}
The last equality follows from Proposition \ref{prop:projection_map_equivariant}, which shows that $\pi_{\alpha}$ is equivariant for any fixed $\alpha$.
\end{remark}
\subsection{Further Optimization via Gradient Descent}
A projection $\alpha \in V_n(\mathbb{R}^N)$ is considered to be optimal if it minimizes the cost function \eqref{eq:nuclear_norm_opt}, defined as the sum over data points $y \in Y$ of the squared distance between $y$ and its image under the map $\pi_\alpha$. We approximate a solution to this minimization problem by initializing our guess with $\alpha_{PCA}$, then using the software package Pymanopt \cite{JMLR:v17:16-177} to perform gradient descent. This output is called $\alpha_{GD}$. While $\alpha_{GD}$ is sometimes quite close to $\alpha_{PCA}$, Section \ref{sec:experiments} details a number of examples where $\alpha_{GD}$ improves upon $\alpha_{PCA}$.
The package Pymanopt implements algorithms for manifold optimization that are intrinsic to the manifold.
We now show how we manipulate our cost function into one that Pymanopt can handle.

\begin{proposition}\label{prop:rewriteopt}
The following two optimization problems are equivalent:
\begin{equation}
    \arg\min_{\alpha \in V_n(\RR^N)} \frac{1}{|\mathcal{Y}|}\sum_{y \in \mathcal{Y}} \|y - \pi_\alpha(y)\|_\mathrm{F}^2  = \arg\max_{\alpha \in V_n(\RR^N)} \frac{1}{|\mathcal{Y}|}\sum_{y \in \mathcal{Y}} \|\alpha^\top y\|_*. \quad \label{eq:nuclear_norm_opt}
\end{equation}
\end{proposition} 
Here, we implicitly take $\arg\min$ and $\arg\max$ over the set of $\alpha \in V_n(\RR^N)$ for which $\pi_\alpha(y)$ is defined for each $y \in \mathcal{Y}$. In practice (see Section \ref{sec:experiments}), we have not encountered a case where the optimization problems above cannot be solved with our gradient descent method.
\begin{proof} 
Proof follows from considering the singular value decompositions of $\alpha^\top y$ and $\hat{y_\alpha}$ in

\begin{eqnarray}
    \hat{\alpha}  &=& \arg\min_{\alpha \in V_n(\RR^N)} \frac{1}{|\mathcal{Y}|}\sum_{y \in \mathcal{Y}} \|y - \pi_\alpha(y)\|_\mathrm{F}^2 = \arg\max_{\alpha \in V_n(\RR^N)} \frac{1}{|\mathcal{Y}|}\sum_{y \in \mathcal{Y}}  \langle \alpha^\top y, \hat{y_\alpha}\rangle\\
       &=& \arg\max_{\alpha \in V_n(\RR^N)}\frac{1}{|\mathcal{Y}|}\sum_{y \in \mathcal{Y}} \tr \left( 
\Sigma_{y} 
 \right) = \arg\max_{\alpha \in V_n(\RR^N)} \frac{1}{|\mathcal{Y}|}\sum_{y \in \mathcal{Y}}\|\alpha^\top y\|_*. 
\end{eqnarray} 
\end{proof} 

The nuclear norm is not differentiable, but it does have a subgradient, which is used as the ``gradient'' for many first-order optimization methods.

\subsection{\texorpdfstring{$\alpha_{PCA}$}{a\_PCA} is a critical point when data comes from a low-dimensional manifold}\label{subsec_PCA_critical_point}

If the considered data $\mathcal{Y} \subset V_k(\RR^N)$ lies precisely in the image $\alpha(V_k(\RR^n))$ of an embedded lower-dimensional Stiefel manifold, then $\alpha_{PCA}$ (see Definition \ref{def:pca}) is a critical point of the optimization problem defined in \eqref{eq:nuclear_norm_opt}. That is, we temporarily work under the following assumption that there is no high-dimensional noise, so the columns of the $y$ span a proper subspace of $\RR^N$ with dimension $n$.

\begin{assumption} \label{assumption:data}
For a fixed $A \in V_n(\RR^N)$, for all $y \in \mathcal{Y} \subset V_k(\RR^N)$, there exists an $x\in \mathcal{X} \subset V_k(\RR^n)$ such that $y = A(x)$.
\end{assumption}

\begin{lemma} \label{lem:pca_critical}
Under Assumption \ref{assumption:data},  $\alpha_{PCA}$ (Definition \ref{def:pca}) is a critical point of the optimization problem in \eqref{eq:nuclear_norm_opt}.
\end{lemma}

\begin{proof}
We evaluate the gradient directly. The Euclidean subgradient of \eqref{eq:nuclear_norm_opt} at any $\alpha \in V_n(\RR^N)$ is
\begin{equation}
   \frac{\partial}{\partial \alpha} f(\alpha) = \frac{\partial}{\partial \alpha}\frac{1}{|\mathcal{Y}|}\sum_{y \in \mathcal{Y}}\|\alpha^\top y\|_*=\frac{1}{|\mathcal{Y}|}\sum_{y \in \mathcal{Y}}\frac{\partial}{\partial \alpha}\|\alpha^\top y\|_*  = \frac{1}{|\mathcal{Y}|}\sum_{y \in \mathcal{Y}}  y\hat{y}_\alpha^\top.
\end{equation}
We project that gradient on to $V_{n}(\RR^N)$, obtaining
\begin{align}
&\frac{1}{2} \alpha \cancelto{0}{ \left(\alpha^\top \frac{\partial}{\partial \alpha} f(\alpha) - \left(\alpha^\top \frac{\partial}{\partial \alpha} f(\alpha)\right)^\top\right)} +  (I_N-\alpha \alpha^\top)\frac{\partial}{\partial \alpha} f(\alpha) \label{eq:gradient_full}\\
= & \frac{1}{|\mathcal{Y}|}\sum_{y \in \mathcal{Y}}(I_N-\alpha \alpha^\top)y\hat{y}_\alpha^\top =(I_N-\alpha \alpha^\top) \frac{1}{|\mathcal{Y}|}\sum_{y \in \mathcal{Y}}y\hat{y}_\alpha^\top ,\label{eq:gradient}
\end{align}
where the first term of \eqref{eq:gradient_full} is 0 since the following matrix is symmetric:
\begin{equation}
    \alpha^\top \frac{\partial}{\partial \alpha} f(\alpha) = (\alpha^\top y)\hat{y}_\alpha^\top = (P \Sigma Q^\top)(PQ^\top)^\top = P\Sigma P^\top.
\end{equation}
The columns of $y\hat{y}_\alpha$ are linear combinations of the columns of $y$. Left multiplication by $I_N-\alpha\alpha^\top$ projects those columns to the subspace that is orthogonal to the column space of $\alpha$. Because $y$ is in the column space of $\alpha_{PCA}$ for all $y \in \mathcal{Y}$ under Assumption \ref{assumption:data}, the Riemannian gradient is zero when evaluated at $\alpha = \alpha_{PCA}$.
\end{proof}

\begin{theorem} \label{thm:pca_global_max}
Under Assumption \ref{assumption:data}, $\alpha_{PCA}$ (Definition \ref{def:pca}) is a global maximizer of \eqref{eq:nuclear_norm_opt}.
\end{theorem}

\begin{proof}

By assumption $\alpha^\top y=\alpha^\top Ax$. 
The nuclear norm of the $n\times k$ matrix $\alpha^\top A x$ is bounded above by $k$ times the operator norm, since the operator norm is the largest singular value. Moreover, the operator norm is submultiplicative, so
\begin{equation}
   \frac{1}{|\mathcal{Y}|}\sum_{y \in \mathcal{Y}} ||\alpha^\top y||_* \leq \frac{1}{|\mathcal{X}|}\sum_{x \in \mathcal{X}}  k ||\alpha^\top Ax||_{op}  \leq \frac{1}{|\mathcal{X}|}\sum_{x \in \mathcal{X}}  k ||\alpha^\top||_{op}||A||_{op}||x||_{op}.
\end{equation}
But each of $\alpha^\top, A, y$ are elements of Stiefel manifolds and so their singular values are all 1. So the cost function is bounded above by $k$, and since the case $\alpha=A$ realizes this upper bound, it is the global maximum.

For any output $\alpha$ of PCA on $\mathcal{Y}$ under these assumptions, we must have $A = \alpha g$ for $g\in O(n)$, since they span the same column space. In which case
\begin{equation*}
||A^\top(y)||_* = ||(\alpha g)^\top y||_* = ||g^\top\alpha^\top y||_* = ||\alpha^\top y||_*,
\end{equation*}
where the last line holds because any SVD decomposition $U\Sigma V$ for $\alpha^\top y$ induces one for $g^\top \alpha^\top y$ by $g^\top U\Sigma V$, so they have the same singular values. So any output of PCA is a maximizer of the cost function. 
\end{proof}

The following is a corollary of Lemma \ref{lem:pca_critical} and the last paragraph of the previous proof.

\begin{corollary}
Let $\beta \in V_n(\mathbb{R}^N) \subset \mathbb{R}^{N \times n}$ be any point on the Stiefel manifold such that its columns span the same subspace as $\alpha_{PCA}$ (Definition \ref{def:pca}). In other words, $\beta$ and $\alpha_{PCA}$ map to the same point on the Grassmanian, and $\beta=\alpha_{PCA} g$ for some $g \in O(k)$. Then under Assumption \ref{assumption:data}, $\beta$ is also a critical point and a global maximizer of the optimization problem \eqref{eq:nuclear_norm_opt}. This follows directly from the proof of Lemma \ref{lem:pca_critical} and Theorem \ref{thm:pca_global_max}.
\end{corollary}

Lemma \ref{lem:pca_critical} and Theorem \ref{thm:pca_global_max} together show that $\alpha_{PCA}$ is a critical point and a global optimizer of \eqref{eq:nuclear_norm_opt} under Assumption \ref{assumption:data}. This is a nontrivial result given that the optimization problem over $V_n(\mathbb{R}^N)$ is not convex. Therefore, we choose $\alpha_{PCA}$ as the initialization point for the  gradient descent step in PSC, with the understanding that if the data happens to come from exactly $V_k(\mathbb{R}^n)$, the gradient descent output $\alpha_{GD}$ will be identical to $\alpha_{PCA}$.

\section{Numerical Experiments}\label{sec:experiments}
In this section, we conduct several experiments on synthetic (Sections \ref{sec:low-dim-visualization}, \ref{sec:pga-comparison}, \ref{sec:stimulus-space-model}, \ref{sec:gdnonlinear}) and real-world data (Sections \ref{sec:brain-connectivity-matrices}, \ref{sec:video-clustering}) to showcase the efficacy of our approach. First, section \ref{sec:low-dim-visualization} visually explains PSC using a low-dimensional example. Section \ref{sec:pga-comparison} extensively compares PSC with principal geodesic analysis (PGA), which is the closest method to ours to date. Section \ref{sec:stimulus-space-model} creates a simulated neuroscience experiment to compare PSC to multidimensional scaling (MDS) and persistent cohomology, as well as compare the use of $\alpha_{GD}$ versus $\alpha_{PCA}$ and a random $\alpha$. Section \ref{sec:brain-connectivity-matrices} applies PSC to brain connectivity matrices clustering with $k=1$, and discusses how it can be viewed as PSC on Grassmanian manifolds. Section \ref{sec:video-clustering} demonstrates application of PSC on video clustering with $k>1$. Lastly, Section \ref{sec:gdnonlinear} gives examples where gradient descent outperforms PCA when the data is nonlinearly generated.
Our Python code is freely available at \url{https://www.github.com/crispfish/PSC}.

\subsection{Low-dimensional visualization}
\label{sec:low-dim-visualization}We start by illustrating our algorithm using a low-dimensional example where $N=3, n=2, k=1$. First, we collect 100 noisy data samples $\mathcal{Y} \subset V_{k}(\mathbb{R}^N)$ in the following manner. One hundred points $\mathcal{X} \subset V_{k}(\mathbb{R}^n)$ are selected from the uniform distribution on $V_k(\mathbb{R}^n)$ \cite{chikuse2012statistics}, then mapped to $V_{k}(\mathbb{R}^N)$ using an $\alpha \in V_{n}(\mathbb{R}^N)$ again from the uniform distribution. In this setting, $\alpha$ embeds the unit circle $S^1\subset \mathbb{R}^2$ as a great circle in the unit sphere $S^2\subset \mathbb{R}^3$. 

If desired, we may perturb each point $x$ slightly by some distance $\epsilon>0$ in the following way. For each $x$ we randomly sample a unit vector $u\in \epsilon\cdot S^{N-1}\subset \mathbb{R}^N$, add this to $x$, and project this perturbed vector back on to $V_k(\mathbb{R}^N)$. The resulting points lie in a neighborhood of the image of $\alpha$. See Figure \ref{fig:noisy-sample} for an example. We set $\epsilon=0.8$ in the following experiment.

\begin{figure}[ht]
    \centering
    \begin{subfigure}[b]{0.38\linewidth}
        \includegraphics[width=\linewidth]{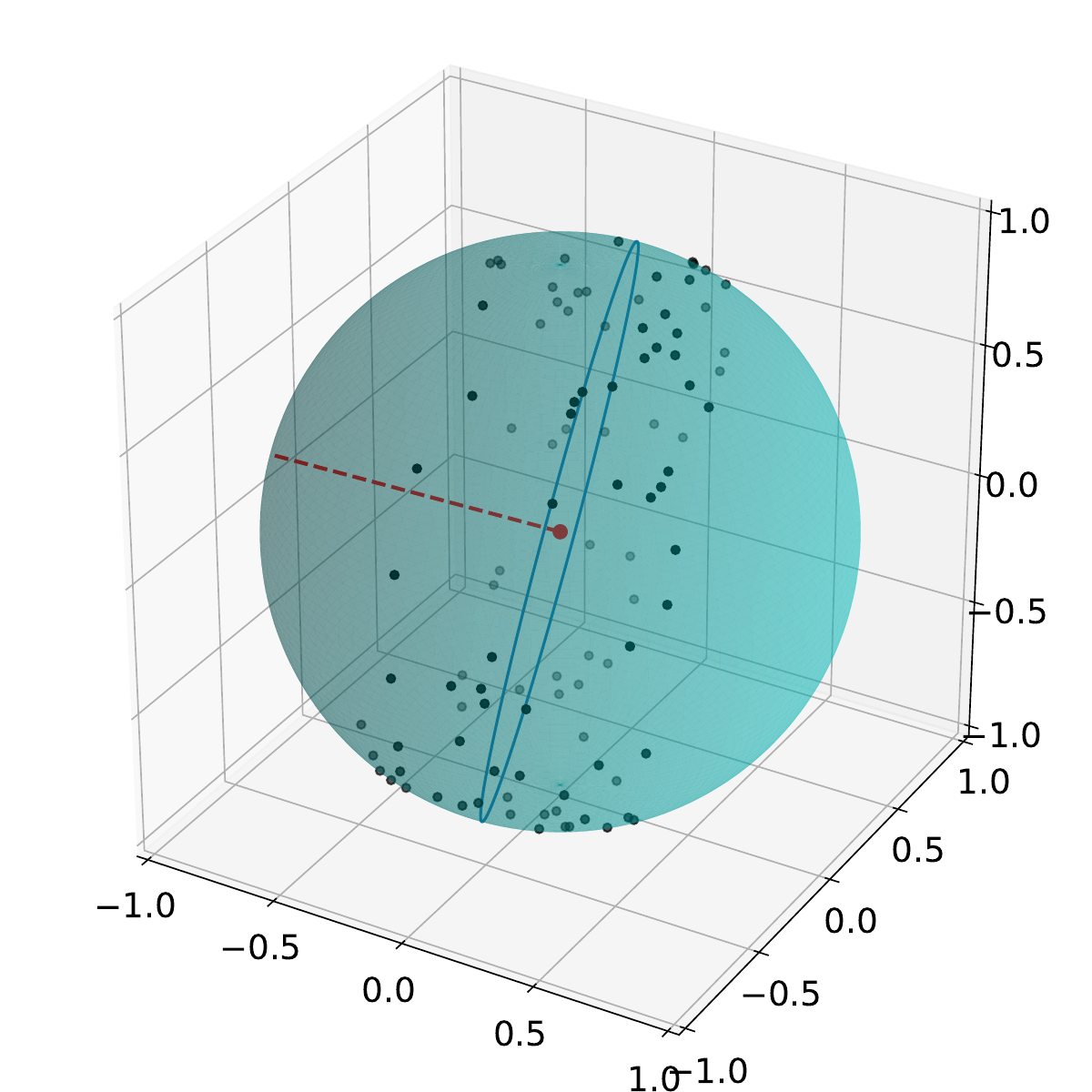}
        \caption{}
        \label{fig:noisy-sample}
    \end{subfigure}
    \hfill
    \begin{subfigure}[b]{0.59\linewidth}
        \includegraphics[width=\linewidth]{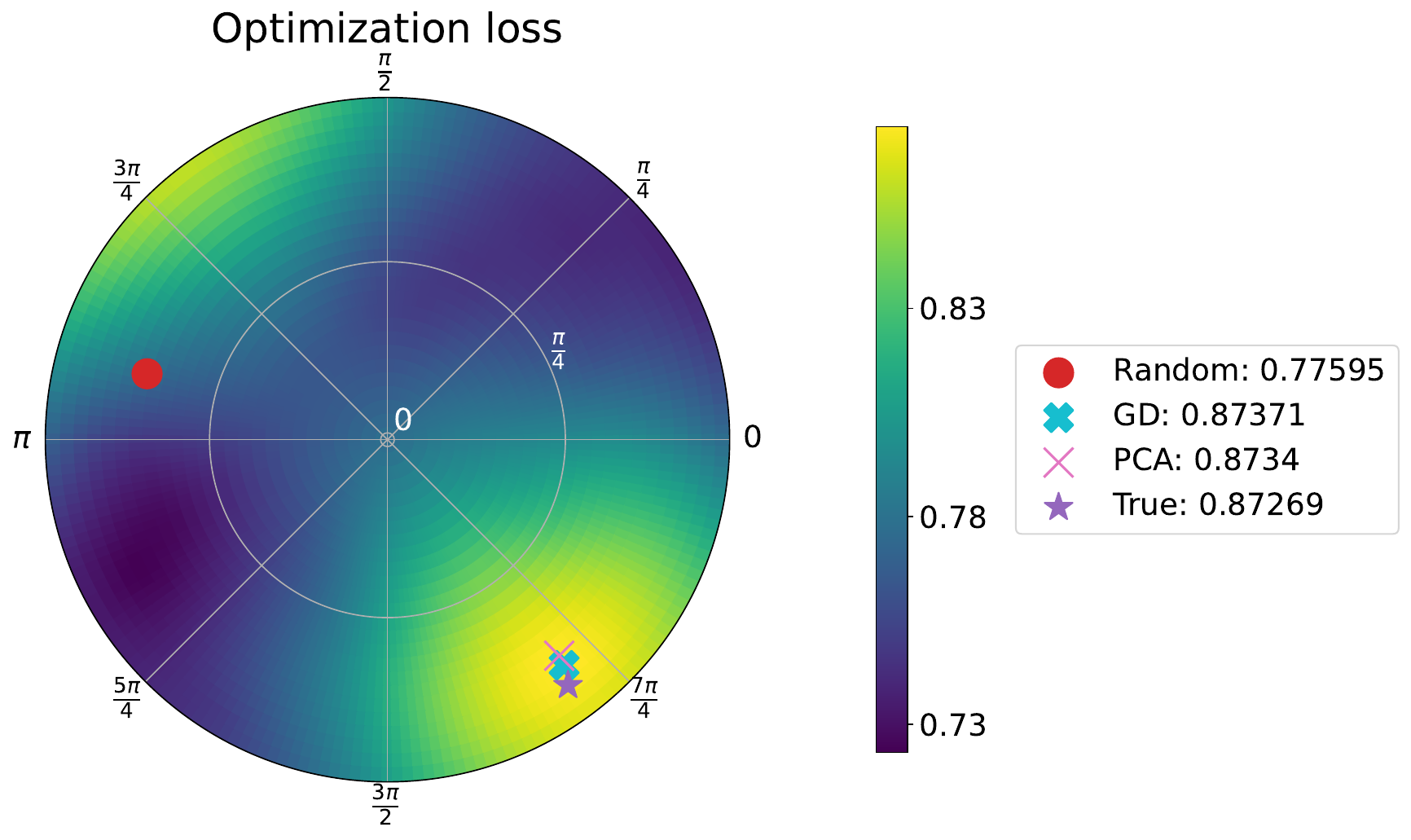}
        \caption{}
        \label{fig:low-dim-sub}
    \end{subfigure}
    \caption{(a) 100 data samples $\mathcal{Y} \subset V_1(\mathbb{R}^3)$ generated with noise level $\epsilon=0.5$. Without noise, all samples lie on the great circle, whose orientation and embedding in the sphere are determined by $\alpha$ as shown in Figure \ref{fig:pipeline_alpha}. (b) Optimization loss landscape \eqref{eq:nuclear_norm_opt} over the projections of $\alpha \in V_{2}(\mathbb{R}^3)$ to the Grassmannian $Gr(1,3) =  \mathbb{R}\textbf{P}^2$. Input data are 100 noisy samples $\mathcal{Y}$ with $\epsilon=0.8$. As described below, the loss function \eqref{eq:nuclear_norm_opt} descends to $Gr(1,3)$ which we plot as the upper hemisphere of $S^2$. The azimuthal angle is $\theta \in [0, 2 \pi]$ and the inclination angle is $\phi \in [0, \frac{\pi}{2}]$ in this polar coordinate plot. Higher function values correspond to lower projection error. ``True" is the $\alpha$ used to generate the noisy samples. 
    }
    \label{fig:low_dim}
\end{figure}

Figure \ref{fig:low-dim-sub} plots the optimization loss landscape \eqref{eq:nuclear_norm_opt} as a function of $\alpha \in  V_{n}(\mathbb{R}^N)$, or more specifically, as a function of the line normal to the column space of $\alpha$. These lines are the span of the red dashed vectors in Figures \ref{fig:pipeline_alpha} and \ref{fig:noisy-sample}. They uniquely define the great circle used to generate the data, and therefore also define $\alpha$ up to an element of $O(2)$. In practice, we calculated these normal vectors simply by taking the cross product of the columns of $\alpha$, and multiplying it by $-1$ if needed to obtain its representative in our model of $\mathbb{R}\textbf{P}^2$. The normal vectors are transformed from Cartesian coordinates $(x, y, z)$ to spherical coordinates using the formula azimuthal angle $\theta = \tan^{-1}(y/x)$ and polar or inclination angle $\phi = \cos^{-1}(z/\sqrt{x^2 + y^2 + z^2})$.

The yellow region in Figure \ref{fig:low-dim-sub} shows the maximum of the loss function here. Note that this region overlaps with the boundary, and continues on the antipodal side of the disc as expected. In this highly noisy setting, $\alpha_{PCA}$ or $\alpha_{GD}$ achieve better performance (i.e. lower projection error) than the actual $\alpha$ (True) that was used in data generation. Lastly, there is a difference between $\alpha_{PCA}$ and $\alpha_{GD}$, showing that the gradient descent algorithm works as it should and is meaningful when the data does not entirely come from a low-dimensional space (e.g. $\epsilon > 0 $).

\subsection{Comparison with PGA} \label{sec:pga-comparison}
In this section, we compare our algorithm to principal geodesic analysis (PGA) as described in \cite{fletcher2004principal}. The package Geomstats \cite{geomstats} was used to apply PGA to data generated synthetically using the process of the previous section. To evaluate the performance of our method, we compare the ratios of appropriately defined variance when projecting using PGA as opposed to PSC.

We define the intrinsic Frechet mean and the Frechet variance as

\begin{equation*}
\mu_{FM}(\mathcal{Y}) = \arg\min_{\mu \in V_k(\RR^N)} \sum_{y\in \mathcal{Y}} D(\mu, y)^2, \quad \sigma^2(\mathcal{Y}) = \frac{1}{|\mathcal{Y}|}\sum_{y\in \mathcal{Y}} D(y, \mu_{FM})^2.
\end{equation*}
It is possible that a mean defined in this way is not unique, especially in the case where the points $S$ lie on a lower-dimensional submanifold with some symmetry.

With variance defined in this way, we consider then the ratio of the variance of the projected points $\pi_\alpha(\mathcal{Y})\subset V_k(\mathbb(\mathbb{R}^n)$ to the variance of the original data points $\mathcal{Y}\subset V_k(\mathbb{R}^N)$. 
Here we set $s = 200, N = 11, n = 6, k = 1$, $\epsilon \in \{0.01, 0.05, 0.1, 0.5\}$, and then apply the procedure described above to generate synthetic data. We then apply both our algorithm and PGA to reduce the dimension, and compute the variance of the projected data.

\begin{figure}[ht]
    \centering
    \includegraphics[width=0.75\linewidth]{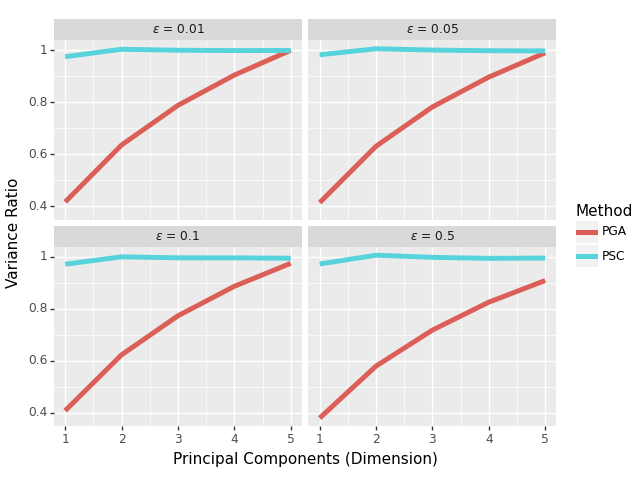}
    \caption{Variance ratio for PSC compared to principal geodesic analysis. Here the $x$-axis is the dimension of the submanifold which is the target of the projection. That is, in dimension 1, PGA projects to a geodesic 1-manifold, while PSC projects to a circle.}
    \label{fig:pga_vs_psf}
\end{figure}

The variance ratio is the quotient of the two variances mentioned above. Each ratio for a fixed dimension of submanifold (either Stiefel in the case of PSC or geodesic in the case of PGA) is shown on the vertical axis of Figure \ref{fig:pga_vs_psf}. Each experiment was performed 100 times, and the results shown in the figure are averaged over 100 iterations.
From this standpoint of expressed variance, our algorithm outperforms PGA for smaller numbers of principal components, capturing a larger amount of the variance contained in the data. Moreover, PSC's advantage over PGA increases as $\epsilon$ increases, demonstrating a robustness to higher levels of noise. 

\subsection{Application to a stimulus space model} \label{sec:stimulus-space-model}
In this section, we discuss an application of our methodology to data arising from neuroscience. Many neural phenomena can be effectively analyzed using a \emph{stimulus space model} \cite{ben-yishaiTheoryOrientationTuning1995,  churchlandNeuralPopulationDynamics2012, deneveReadingPopulationCodes1999, manteContextdependentComputationRecurrent2013, seungSimpleModelsReading1993}. This model supposes the existence of an underlying metric stimulus space $(S, d_S)$ and assigns neurons $\sett{1, 2, \ldots, k}$ in a population to points $\sett{n_1, n_2, \ldots, n_k}$ in $S$. Then, a stimulus is modeled as a time-series $\sett{s_1, s_2, \ldots, s_N}$ of points in $S$. In general, a neuron's response to a stimulus decreases with distance. Thus, the response $r_i(s_k)$ of neuron $i$ to stimulus $s_k$ can be described by
\begin{equation}
    r_i(s_k) = f_i(d_S(n_i, s_k))
\end{equation}
where $f_i$ is some (usually unknown) monotonically decreasing function.

A central problem in modern computational neuroscience is to extract information about the stimulus space or the stimulus itself from the activity of the encoding neurons alone. The fact that the response functions $f_i$ are unknown, highly nonlinear, and often vary with time means that standard linear methods are often of limited use. In many observed systems, the underlying stimulus space has an inherent and nontrivial topology. For instance, head direction \cite{chaudhuriIntrinsicAttractorManifold2019, kimRingAttractorDynamics2017, turner-evansAngularVelocityIntegration2017a} and orientation in visual cortex \cite{ben-yishaiTheoryOrientationTuning1995} can be modeled by circular stimulus spaces. In the case when the stimulus space is fully sampled, persistence techniques can be used to decode topological information about the stimulus \cite{rybakkenDecodingNeuralData2019}. See also \cite{Kang-Xu-Morozov-State_space}, \cite[Section 6]{scoccola_toroidal}, and \cite[Section 6]{schonsheck2023spherical}.

We simulate the case of a circular stimulus space that has been only partially sampled. For instance, while head direction can often be well-described with circular coordinates, limited range of motion (e.g., looking roughly 90 degrees to the left and right) means that the associated circular stimulus space may only be sampled from, for instance, $0$ to $\pi$. In this case, methods as in \cite{rybakkenDecodingNeuralData2019} will be ineffective as they require detecting a high-persistence feature of the data, such as a ``full'' circle, using persistent cohomology.

In further detail, we consider a population $\sett{1, 2, \ldots, 100}$ of 100 neurons identified with uniformly distributed points $\sett{n_1, n_2, \ldots, n_{100}}$ on a half-circle, from $0$ to $\pi$. Each neuron $n_i$ is assigned a response function $  f_i = \max(1 - m_i x, 0)$

where the slope $m_i$ is chosen uniformly randomly between 25 and 50. 
We simulated neural activity in response to a stimulus modeled by a random walk $\{t_1, t_2, \allowbreak \ldots, t_{13,000}\}$ on the half-circle comprising 13,000 steps. Activity of the neural population at step $t_i$ can be recorded as vector with $N=100$ entries, one for each neuron. Thus, the collective activity over the whole time series can be encoded as a point cloud in $\RR^{100}$ with $|\mathcal{Y}| = 13,000$ points.

\begin{figure}[ht]
    \centering
    \includegraphics[width=0.63\linewidth]{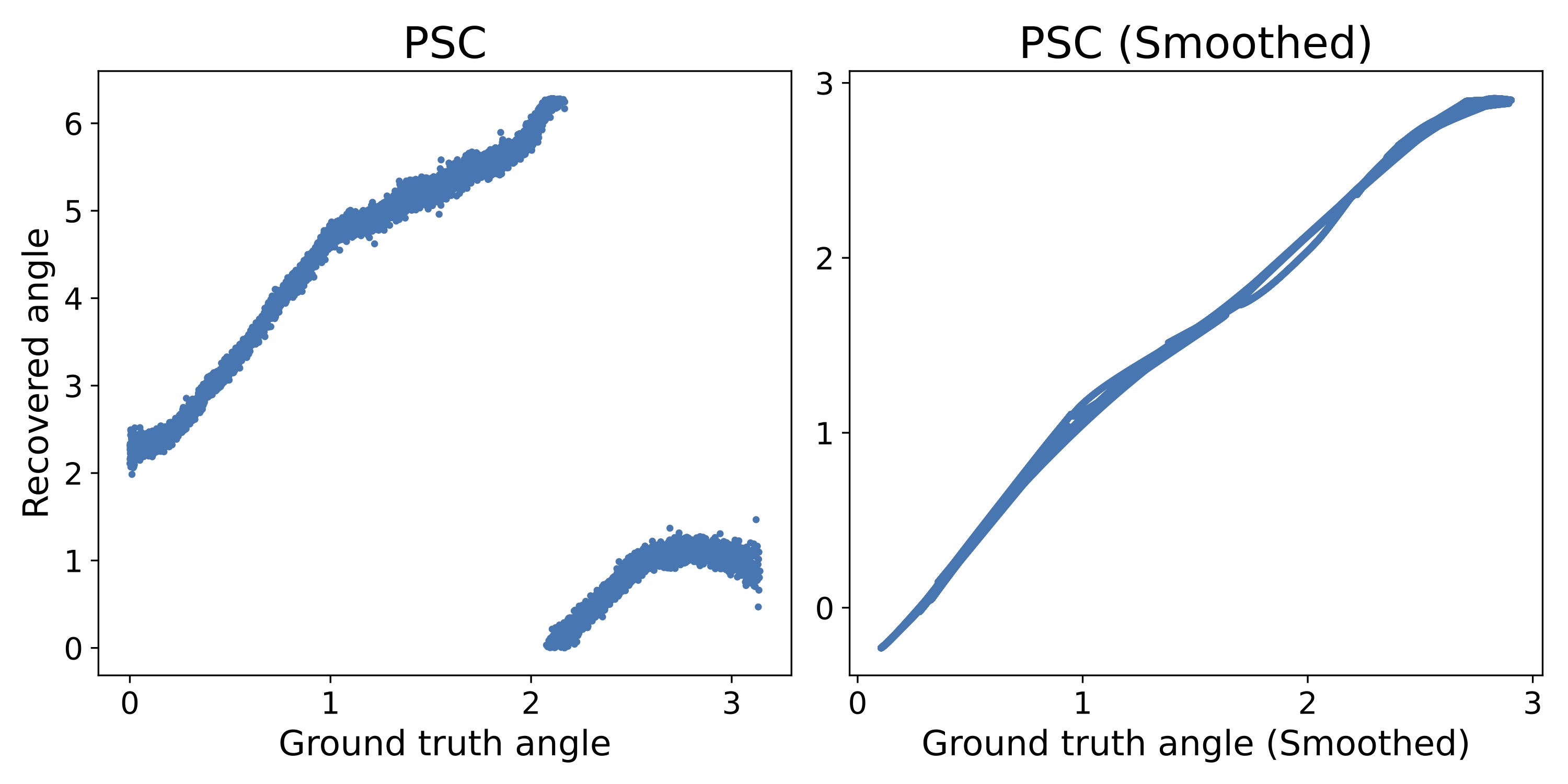}%
    \includegraphics[width=0.35\linewidth]{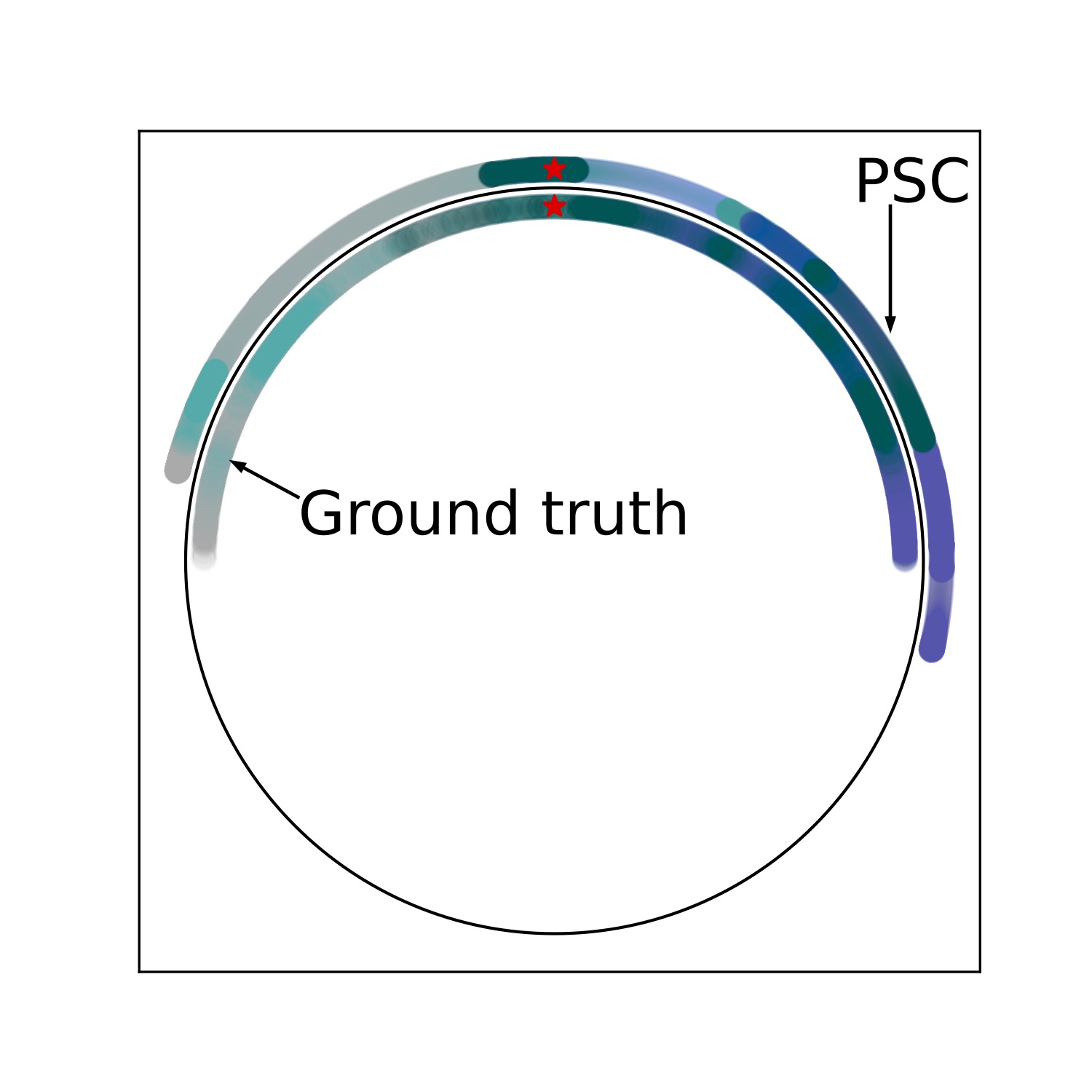}
    \caption{(Left) Angular coordinates of a geometric walk on a half-circle against  coordinates recovered via PSC from neural response. (Middle) Post-processed and smoothed recovered coordinates against smoothed ground truth angles. (Right)  Angles plotted on a circle. Smoothed ground truth angles are plotted in the inner circle, while the post-processed and smoothed PSC angles are in the outer ring. PSC angles are shifted to align with ground truth angle at $t=0$ (red stars). Color intensity correlates with time steps. }
    \label{fig_neuro_ex_coord_recovery_with_pipeline}
\end{figure}

First, we centered the data by subtracting the neural population's average response from each entry. Then, we normalized the data to obtain a point cloud on $S^{99}$. Hypothesizing that the neural stimulus (i.e., random walk) could be well-approximated by a time series on a half-circle, we ran Algorithm \ref{algorithm_stiefel} with $n = 2$. In Figure \ref{fig_neuro_ex_coord_recovery_with_pipeline}, we compare the recovered data on $S^1$ to the ground truth coordinates of the random walk as in \cite[Section 3]{desilvaPersistentCohomologyCircular2011}.

The takeaway of Figure \ref{fig_neuro_ex_coord_recovery_with_pipeline} is that our method accurately recovers the neural stimulus up to rotation and scaling, each of which can easily be addressed in post-processing with knowledge of the spatio-temporal relationship of a few points of the ground truth input. Indeed, after 1) uniformly rotating the recovered coordinates so that the recovered coordinate of the initial time step matched that of the ground truth, and 2) scaling the recovered coordinates on the circle, we recover the stimulus (i.e., random walk on the half-circle) almost exactly. Figure \ref{fig_neuro_ex_recovered_time_series} shows the recovered walk with no post-processing (top right) and with post-processing (bottom right) as well as the ground truth walk with and without Gaussian smoothing (top and bottom left, respectively). Note that the y-axis in the upper right-hand plot is ``circular,'' i.e., $2\pi$-periodic, as it represents an angle on a circle and accounts for the apparent discontinuity in the plot. This periodicity also means that points near the top of the plot are, in fact, mapped quite close to points near the bottom, which explains the spurious point at approximately (12500, 0.25).

\begin{figure}[ht]
\begin{subfigure}[t]{0.5\linewidth}
\centering
\fbox{\includegraphics[width=0.95\linewidth]{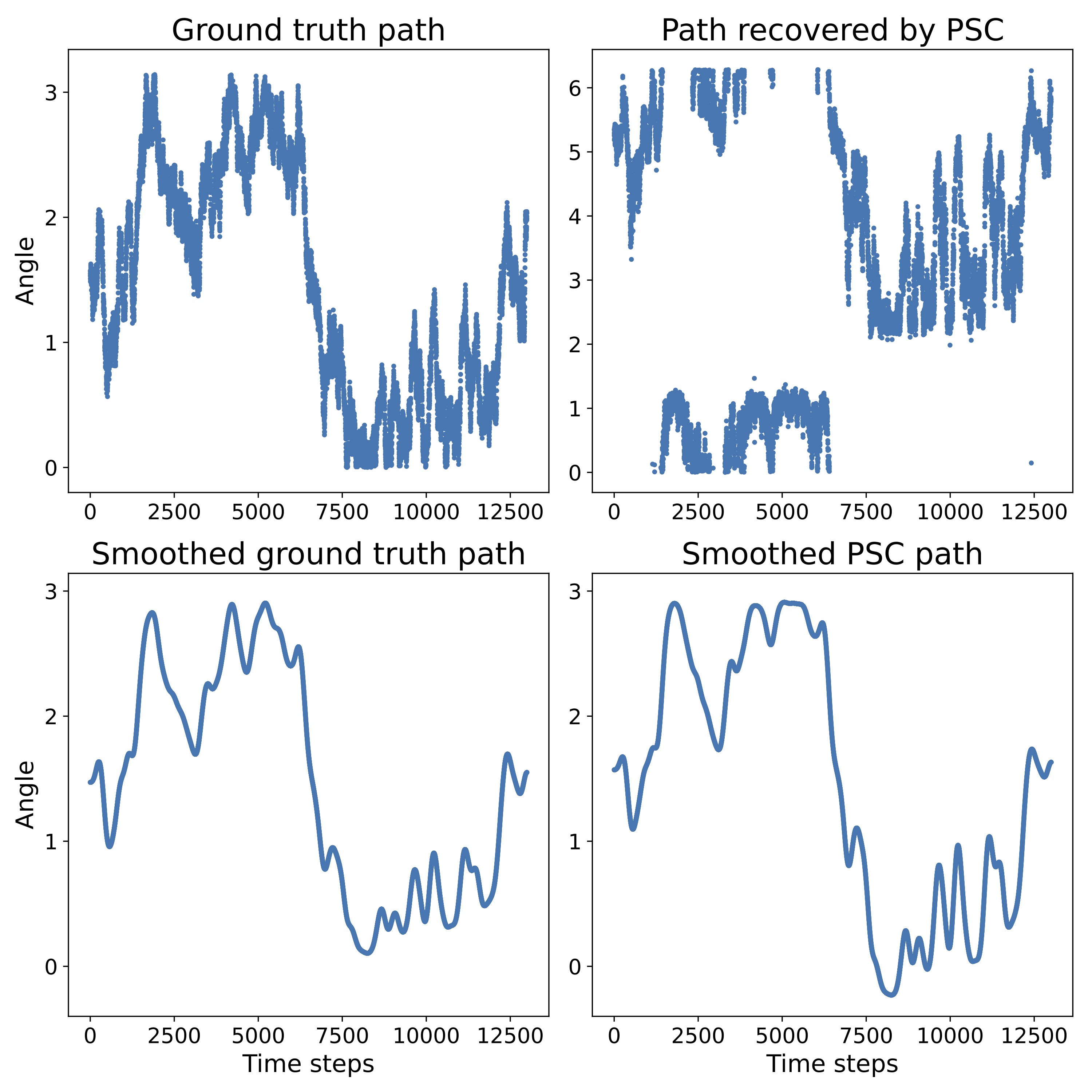}}
\caption{Ground truth vs. PSC}
\end{subfigure}%
\begin{subfigure}[t]{0.5\linewidth}
\centering
\fbox{\includegraphics[width=0.95\linewidth]{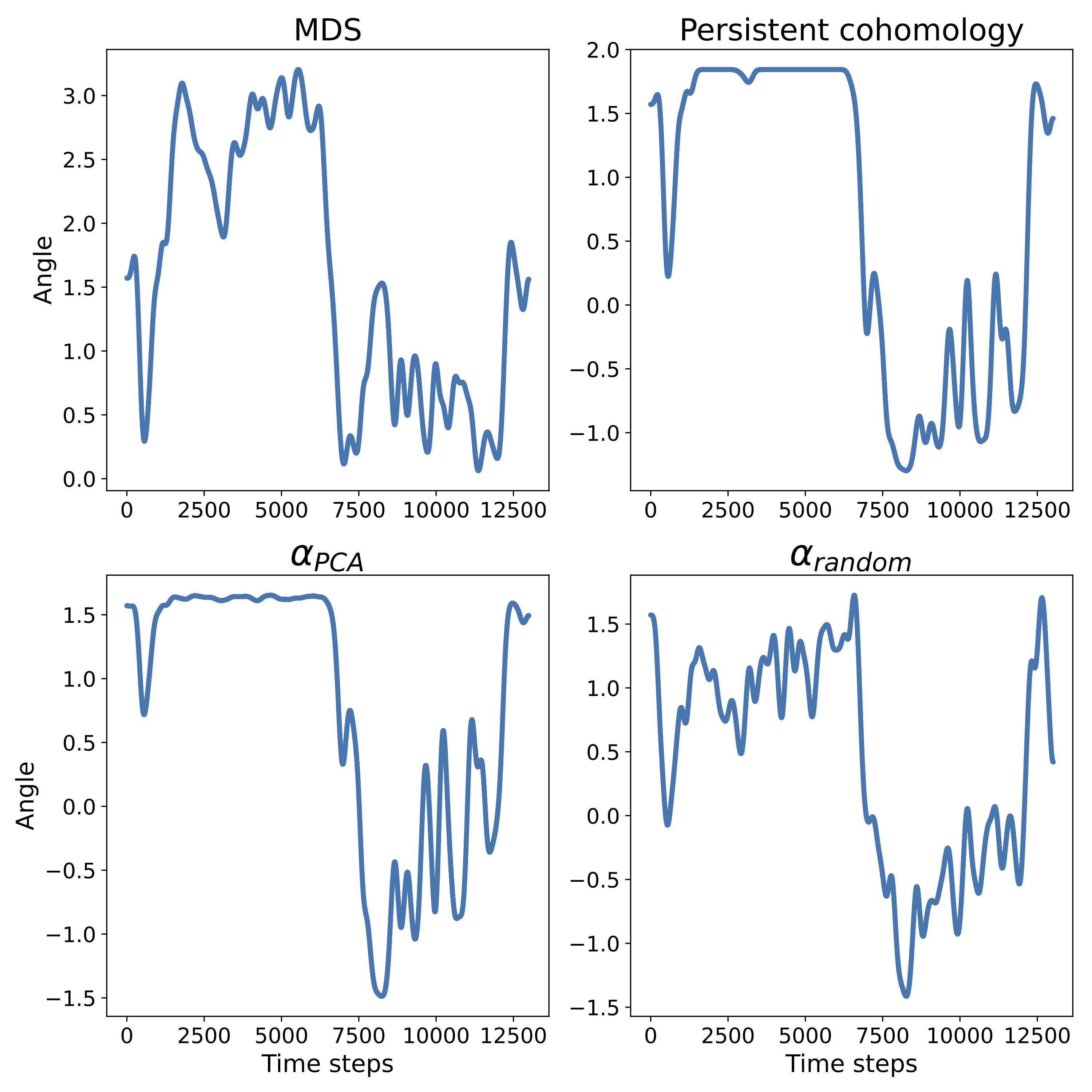}}
\caption{Other methods}
\end{subfigure}
    \caption{(a) Ground truth recovered time series (top left), recovered time series with no post-processing (top right), smoothed ground truth time series (bottom left), and smoothed, post-processed recovered time series (bottom right). MSE between post-processed and smoothed paths is 0.027 (rounded up to 3 decimal points). (b)  Recovered paths after post-processing and smoothing from MDS (top left), persistent cohomology (top right), with $\alpha_{PCA}$ (bottom left), and with a random $\alpha$ (bottom right). MSEs are $\alpha_{PCA} = 0.655, \alpha_{random} = 1.212$, MDS = $0.199$, and persistent cohomology = $0.770$.}
\label{fig_neuro_ex_recovered_time_series}
\end{figure}

By contrast, existing methods such as multidimensional scaling or persistent cohomology fail to accurately recover the input stimulus from the recorded neural response. Figure \ref{fig_neuro_ex_recovered_time_series} shows the post-processed and smoothed paths from using multidimensional scaling (top left, limited to 150 iterations to avoid excessive computation time) and persistent cohomology (top right). For the persistence approach, we computed the persistent cohomology of the 100-dimensional point cloud and used the circular coordinates algorithm of \cite{pereaSparseCircularCoordinates2020} to extract a map from the data to $S^1$ from the longest-lived (though still very short) persistent cohomology class. Lastly, this example also highlights the importance of choosing the ``right'' $\alpha$ as described in Algorithm \ref{algorithm_stiefel}. Figure \ref{fig_neuro_ex_recovered_time_series} also shows coordinate recovery using a random $\alpha$ (bottom left) and $\alpha_{PCA}$ (bottom right). Using a random $\alpha$ clearly yields poor recovery and while $\alpha_{PCA}$ accurately recovers some points, there is notable improvement using $\alpha_{GD}$. 

\subsection{Application to brain connectivity matrices} \label{sec:brain-connectivity-matrices}
Next, we show how our algorithm can be applied to real-world data. In many neuroscience applications, brains are divided into $N$ regions before their activities are measured by, for example, fMRI.  Positive semidefinite (PSD) matrices in $\mathbb{R}^{N \times N}$ naturally arise in this context as graph adjacency matrices, graph Laplacian matrices, network connectivity matrices, and partial correlation matrices \cite{mantoux2021understanding, muldoon2016stimulation, vskoch2022human}. 
As eigenvectors contain important information about these matrices \cite{atasoy2016human,chen2020learning}, it is intuitive to represent the set of top $k$ eigenvectors as an orthonormal $k$-frame and thus as a point in the Stiefel manifold  $V_k(\mathbb{R}^N)$. They can then be used for various downstream tasks such as clustering and brain-computer-interface control, for which dimensionality reduction is useful.

In this section, we demonstrate our algorithm on pairwise functional connectivity matrices between $N=83$ brain regions, using data from \cite{muldoon2016stimulation}. This is derived from three scans from eight subjects, resulting in a dataset of 24 PSD matrices in $\mathbb{R}^{83 \times 83}$. From each matrix, we took $k=1$ eigenvector corresponding to the largest eigenvalue, and for each of the resulting 24 data points on $V_1(\mathbb{R}^{83})$, we projected them to $V_1(\mathbb{R}^{3})$ using our proposed algorithm. The number $n=3$ is chosen mainly for visualization, and $k=1$ is chosen based on ease of visualization and \cite{mantoux2021understanding}. However, $k > 1$ is also commonly used in neuroscience depending on the scientific problem; for example, $k=60$ in \cite{chen2020learning}. 

\begin{figure}[ht]
    \centering
    \includegraphics[width=0.3\linewidth]{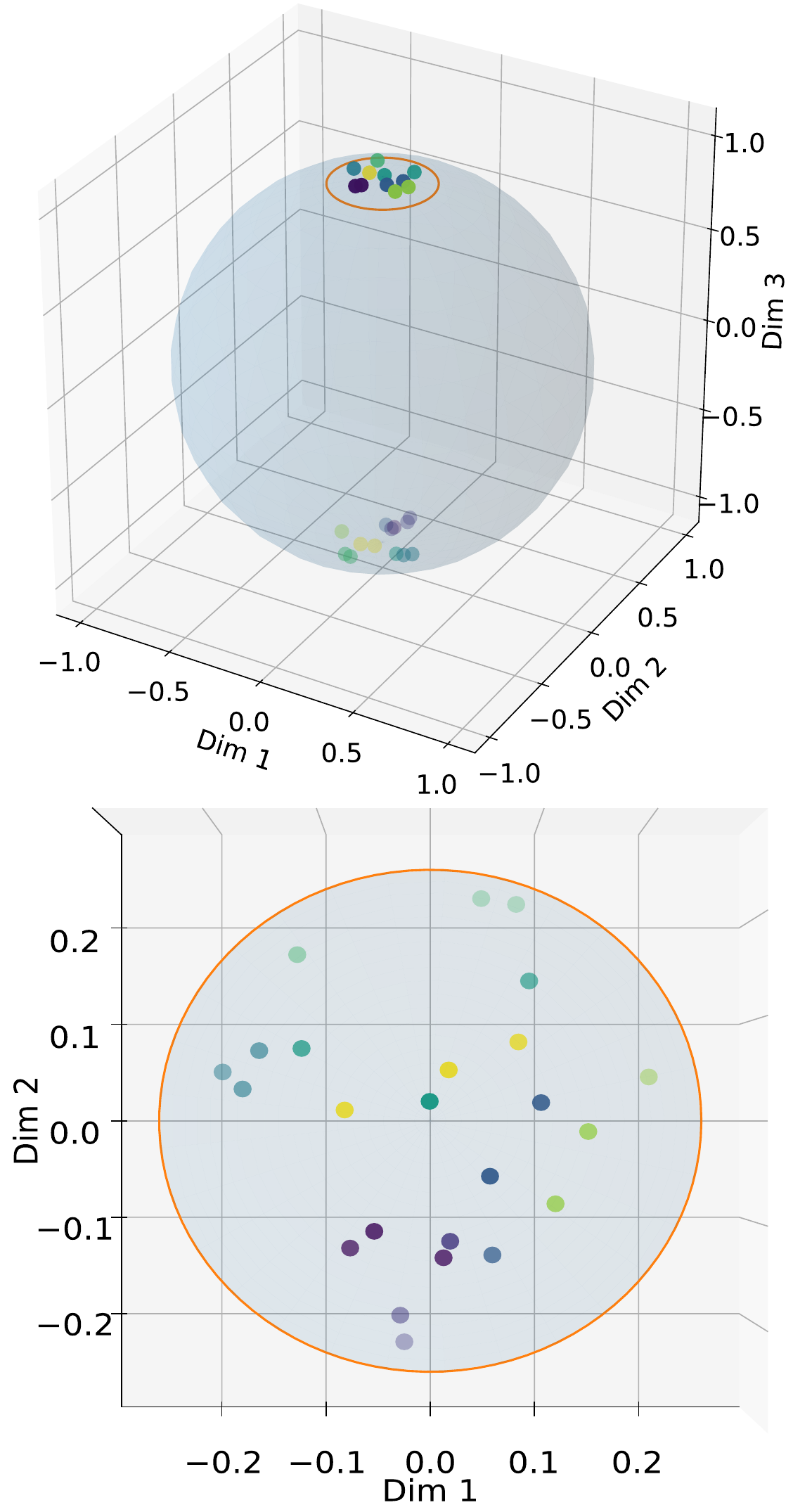}
    \includegraphics[width=0.66\linewidth]{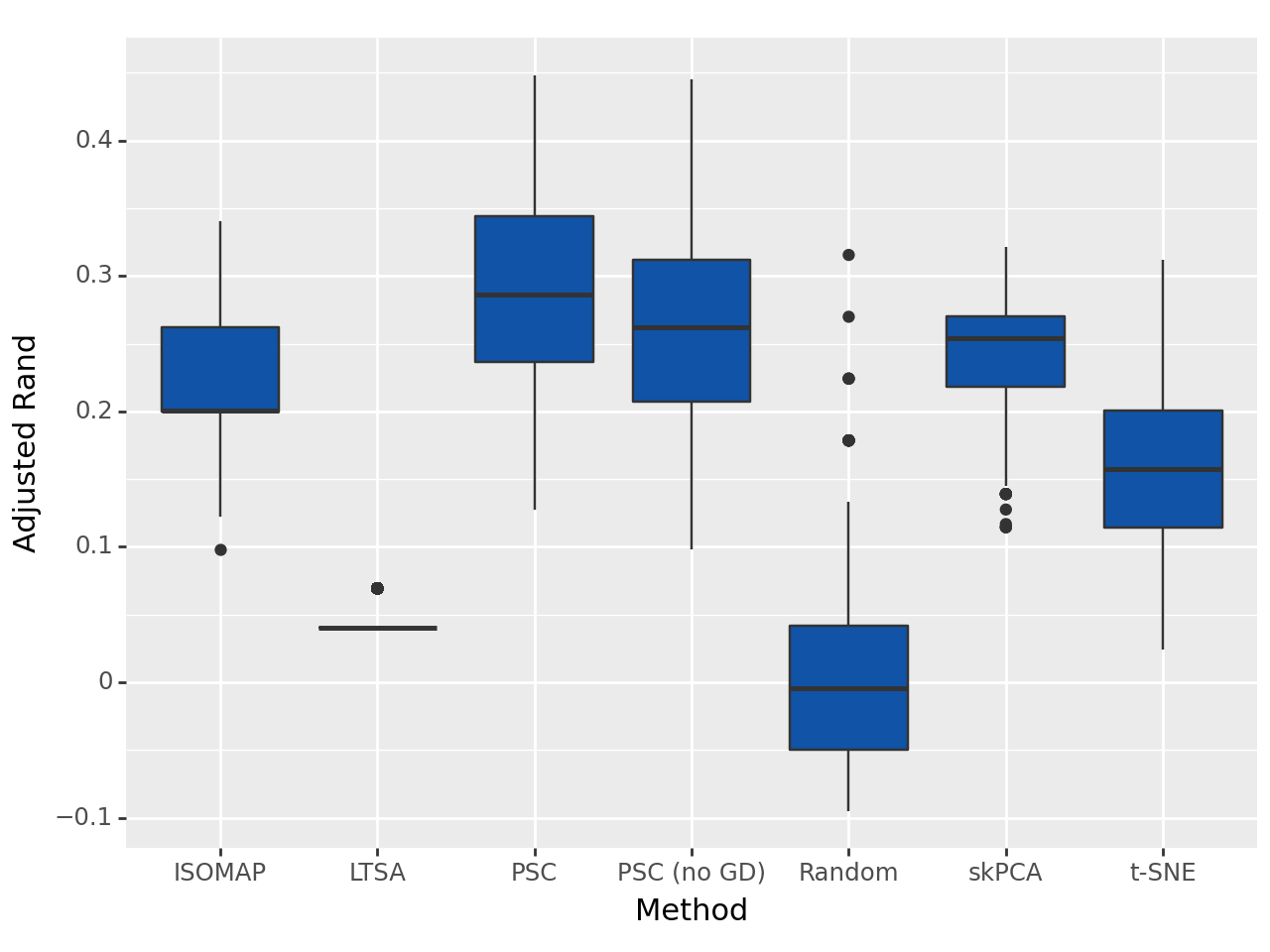}
    \caption{(Top left) 24 brain connectivity matrices projected from $\mathbb{R}^{83\times 83}$ to $V_1(\mathbb{R}^{83})$ to $V_1(\mathbb{R}^{3})$, colored by their corresponding eight human subjects. (Bottom left) Zoomed-in and bird's-eye view of the data after all points in the lower hemisphere are reflected across the $z=0$ plane to the upper hemisphere. (Right) A boxplot from the results of a test comparing adjusted Rand index of Riemannian k-means clustering on $V_1(\mathbb{R}^3)$ to other dimensionality reduction methods and random label assignment. PSC with gradient descent outperforms other methods.}
    \label{fig:brain_projected}
\end{figure}

On the left side of Figure \ref{fig:brain_projected}, we plot the 24 data points projected onto $V_1(\mathbb{R}^{3})$. They are colored according to the eight human subjects. The right plot in Figure \ref{fig:brain_projected} shows the zoomed-in and bird's-eye view of the data after all points in the lower hemisphere are reflected across the $z=0$ plane to the upper hemisphere. Samples from the same human subjects seem to sit closer to each other, which is remarkable given the amount of information that was lost when dimensionality was reduced from $3,403=(83^2-83)/2$ unique numbers in $\mathbb{R}^{83\times 83}$ to only $2$ unique numbers in $V_1(\mathbb{R}^{3})$. 

To measure PSC's performance on this data set, we applied k-means clustering to the dimensionally reduced data lying on $S^2$, with the goal of recovering the groupings coming from the eight patients appearing in the study. These clusters were then compared to the ground truth using the adjusted Rand index. This measure is bounded below by -0.5 (for clusterings which are worse than the average random assignment) and bounded above by 1 (for clusterings identical to the ground truth).\newline\indent
We compare clustering after applying PSC to other leading dimensionality reduction methods, as well as random label assignment, as follows. Each application of k-means clustering results in a slightly different adjusted Rand index, so we applied the clustering algorithm 1000 times to obtain a Monte Carlo approximation for the distribution of scores. A boxplot of these scores for each dimensionality reduction method appears in Figure \ref{fig:brain_projected}.  The adjusted Rand score for PSC clearly outperforms random label assignment, preserving relevant information even after drastic reductions in both the number of principal components and the dimension of the Stiefel manifold. Moreover, PSC with gradient descent outperforms the other methods. These comparisons still hold for other measures of clustering performance, e.g. Fowlkes-Mallows, homogeneity, completeness, etc. 

Finally, we remark that the procedure described above can also be understood as an implementation of Algorithm \ref{algorithm_grassmannian}, a dimensionality reduction on Grassmannian manifolds rather than Stiefel manifolds, in the following way. Consider the commutative diagram below, where the vertical maps are the natural quotient maps, $\alpha$ is as in Algorithm \ref{algorithm_stiefel}, and $\bar{\alpha}$ is naturally induced on the associated quotient spaces.

\begin{align}\label{eqn_brain_data_comm_diagram}
    \xymatrix{
    V_1(\RR^3) \ar[d] \ar[r]^-\alpha & V_1(\RR^{83})\ar[d]\\
    Gr(1, \RR^3) \ar[r]^-{\bar{\alpha}} & Gr(1, \RR^{83})\\
    }
\end{align}

For each point $y$ in our data set of 24 PSD matrices, we consider its largest eigenvalue $\lambda(y)$ and associated eigenspace $E_{\lambda(y)}$, which is an element of $Gr(1, \RR^{83})$. Selecting a particular unit-length eigenvector from each eigenspace is well-defined only up to multiplication by $-1$, and constitutes a ``lift'' of $E_{\lambda(y)}$ to an element $\tilde{y}$ of $V_1(\RR^{83})$. Then, we perform Algorithm \ref{algorithm_stiefel} to obtain data in $V_1(\RR^3)$ (left hand of Figure \ref{fig:brain_projected}) and finally project to $Gr(1, \RR^3)$ (right hand of Figure \ref{fig:brain_projected}).

\subsection{Application to video}
\label{sec:video-clustering}

\begin{figure}[ht]
    \centering
    \includegraphics[width=0.4\linewidth]{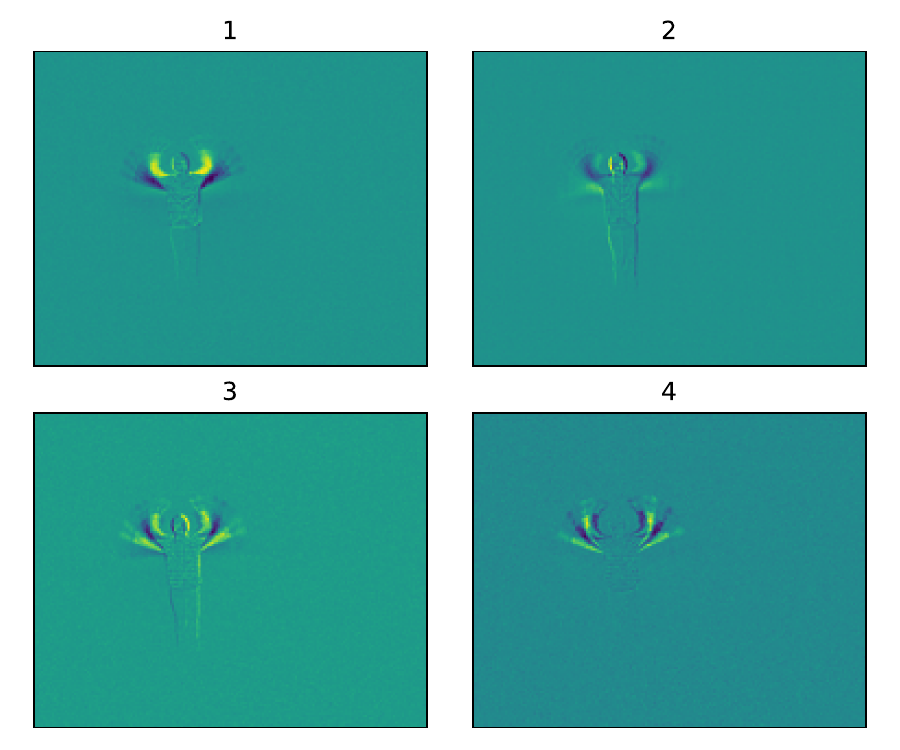}%
    \includegraphics[width=0.4\linewidth]{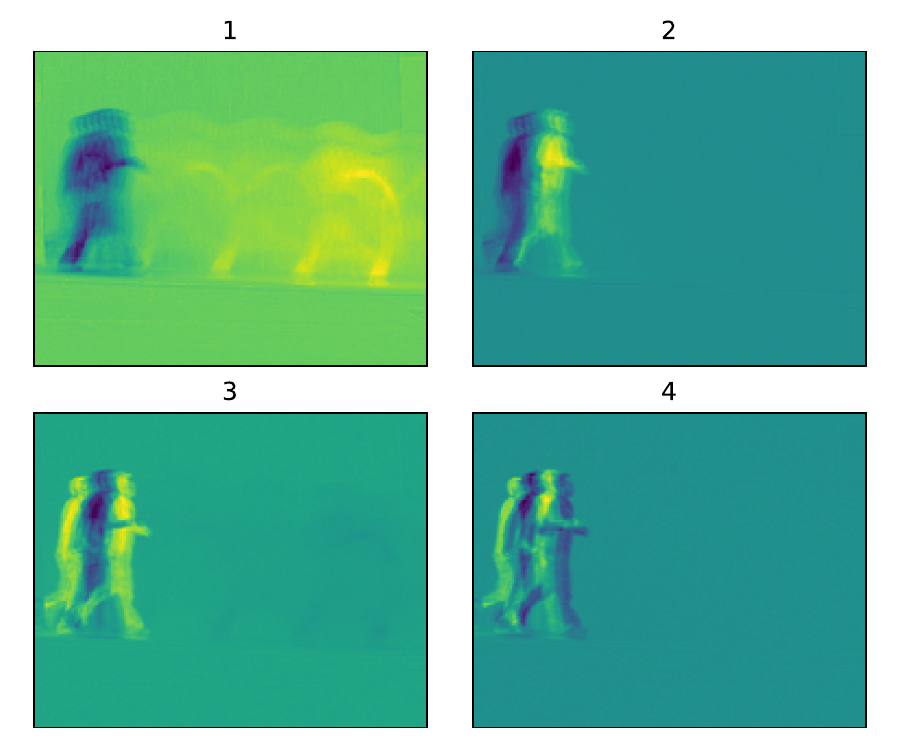}
    \caption{Video clips represented as data points on the Stiefel manifold $V_{4}(\mathbb{R}^{25920})$. Left: Person waving both arms. Right: Person running to the right.}
    \label{fig:videos}
\end{figure}

In this example, we explore how PSC may be applied to video data with $k>1$. Works such as \cite{turaga2011statistical,hayat2014automatic} have studied how video frames can be represented as data points living on Stiefel and Grassmanian manifolds, which can then be used for facial expression classification or clustering. The action database in \cite{gorelick2007actions} has low resolution videos of 144 $\times$ 180 pixels at 50 fps, including one of someone waving both arms, and one of another person running from left to right. Following the processing steps in \cite{gorelick2007actions}, we used a sliding window size of eight frames every four frames to divide up the two videos into 15 and 13 shorter clips. Then, we turned them from RGB to grayscale and subtracted the mean pixel values to remove the stationary background. This yields $|\mathcal{Y}|=28$ video cubes of $\mathbb{R}^{144 \times 180 \times 8}$, which can be flattened in pixel dimensions to size $\mathbb{R}^{25920 \times 8}$. Apply SVD and take the $k=4$ leading left singular vectors as in \cite{hayat2014automatic} to create a dataset $\mathcal{Y}$ living on $V_{4}(\mathbb{R}^{25920})$. Remarkably, this process captures the movements in the video; see Figure \ref{fig:videos}.

To demonstrate a potential downstream task, we applied spectral clustering (as in \cite{hayat2014automatic}), which correctly classified 22 out of 28 clips. Then, we applied PSC on $\mathcal{Y}$ to reduce dimensionality from $N=25920$ to $n=20$, and re-applied spectral clustering. Impressively, the classification results were the same. This shows that while video data live in very high-dimensional ambient space, algorithms such as PSC can successfully embed the discriminative information in a much lower  ($\approx 0.077\%$) dimensional space.

\subsubsection{Choosing \texorpdfstring{$n$}{n} in real-world data}\label{appendix_choosing_n}

In the case of real-world data where the best $n$ might not be known a priori, we recommend using the eigenvalues or singular values that result from deriving $\alpha_{PCA}$ (Definition \ref{def:pca}). There may be jumps or an ``elbow'' indicating a natural cutoff value, similar to how practitioners choose the number of principal components in standard PCA (e.g. Minka's MLE \cite{minkamle2000}). See Figure \ref{fig:singular-values}.

\begin{figure}[th]
    \centering
    \includegraphics[width=0.3\linewidth]{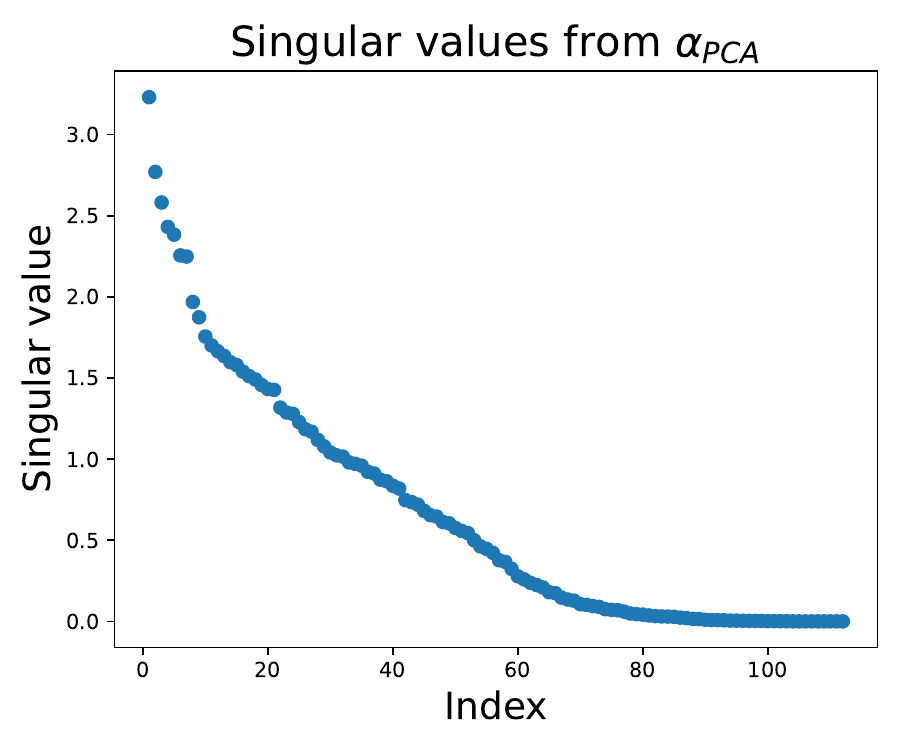}%
    \includegraphics[width=0.3\linewidth]{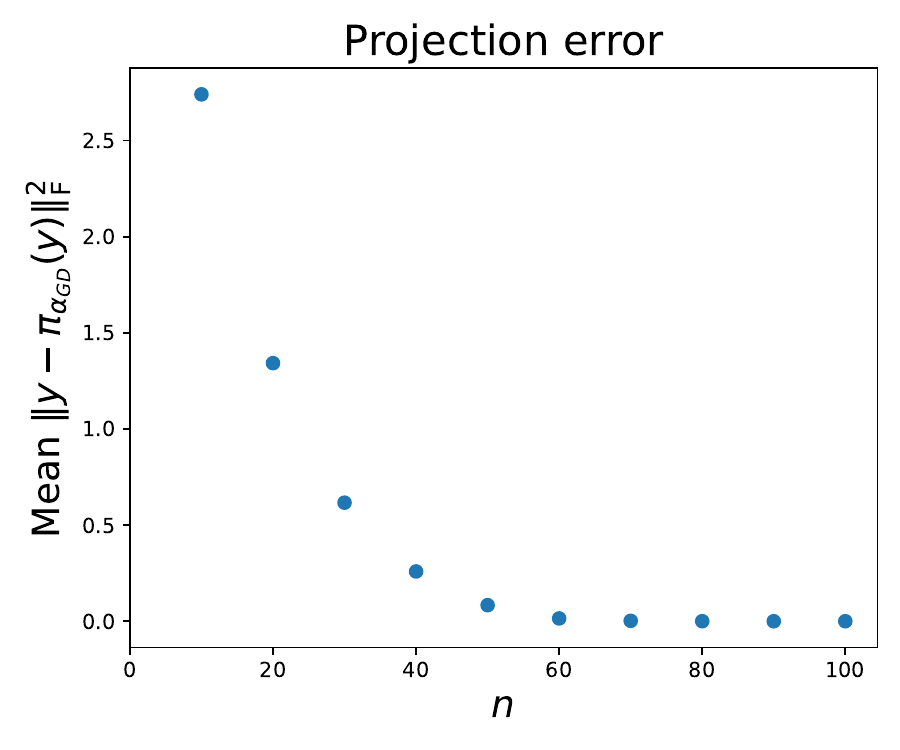}
    \caption{PSC dependence on $n$ in video data. Larger $n$ leads to less projection error, but there is a trade-off in that smaller $n$ can sometimes lead to more robustness or speed in downstream tasks due to stronger denoising or compression behavior.}
    \label{fig:singular-values}
\end{figure}

\subsection{Gradient descent is necessary when data is nonlinearly generated}\label{sec:gdnonlinear}

In most of the previous examples, the difference between $\alpha_{PCA}$ and $\alpha_{GD}$ is relatively minor. A notable exception is in the case of the highly nonlinear embedding of a (half) circle discussed in the stimulus space example of Section \ref{sec:stimulus-space-model}, where $\alpha_{GD}$ significantly outperforms $\alpha_{PCA}$. In the present section, we discuss a family of synthetic data sets arising from vector bundles, similar in spirit to the stimulus space example, in which $\alpha_{GD}$ produces statistically significant improvements in mean squared error over $\alpha_{PCA}$. These examples also illustrate the need for an \textit{equivariant} dimensionality reduction pipeline as data $X$ is discontinuously embedded in $V_k(\mathbb{R}^n)$, and one only obtains a continuous embedding after projection to the relevant Grassmannian manifold.

Here we outline how Stiefel manifold-valued data naturally arises from vector bundles. Examples of real-world data well-modeled by this modality include the study of orientability of attractors in certain dynamical systems \cite[Section 7.2]{Scoccola_Perea_Approx_and_Discrete}, data sets of lines in Euclidean space (e.g., weights learned by various neural networks) \cite[Section 7.3]{Scoccola_Perea_Approx_and_Discrete}, and synchronization in cryogenic electron microscopy \cite[Section 7.4]{Scoccola_Perea_Approx_and_Discrete}. For a thorough treatment of vector bundles, we refer the reader to \cite{Hatcher_VB}. For further details on the method described below, see \cite{perea2018multiscale,pereaSparseCircularCoordinates2020} as well as \cite{Perea_Polanco,  Scoccola_Perea_Approx_and_Discrete, scoccola_Perea_FibeRed}. 

Throughout, let $X \subset M$ be a data set of interest, where $M$ is a paracompact topological space, e.g., a closed manifold. Begin by noting that the inclusions $\mathbb{R}^m \subset \mathbb{R}^{m+1} \subset \cdots$ induce inclusions $Gr(k, \mathbb{R}^m) \subset Gr(k, \mathbb{R}^{m+1}) \subset \cdots$ and let $Gr(k, \mathbb{R}^\infty)$ be the infinite union $\cup_{m}Gr(k, \mathbb{R}^m)$. It is a well-known fact \cite[Theorem 1.16]{Hatcher_VB} that the set of isomorphism classes Vect$^k(M)$ of rank-$k$ vector bundles over $M$ is in bijection with the set of homotopy classes of maps from $M$ to $Gr(k, \mathbb{R}^\infty)$. That is, $\text{Vect}^k(M) \cong [ M, Gr(k, \mathbb{R}^\infty)]$ and so giving the data of an isomorphism class of rank-$k$ vector bundles over $M$ is equivalent to selecting a homotopy class of maps $[g]$ from $M$ to $Gr(k, \infty)$. Thus, given a data set $X \subset M$ and a vector bundle over $M$, one obtains a data set in $Gr(k, \infty)$ by considering the image $g(X)$ of the data under a representative $g$, known as the classifying map of the bundle, of the class $[g]$ specified by this bijection.

Given a rank-$k$ vector bundle over $M$, Perea \cite{perea2018multiscale} details how to construct an explicit (i.e., data-friendly) representation of the corresponding classifying map. In brief, we begin with 1) a trivializing open cover $\mathcal{U} = \{ U_j \}_{j=1}^J$ of $M$ and 2) a partition of unity $\{ \phi_j \}_{j=1}^J$ subordinate to $\mathcal{U}$, and consider the nerve $N(\mathcal{U})$ of the open cover, an abstract simplicial complex with a vertex $v_j$ for each $U_j \in \mathcal{U}$, an edge $v_jv_\ell$ if $U_j \cap U_\ell \neq \emptyset$, a 2-simplex $v_jv_\ell v_m$ if $U_j \cap U_\ell \cap U_m \neq \emptyset$, and so on. Then, we use a representative $\eta$ of a suitable class $[\eta] \in H^k(N(\mathcal{U}); R)$, in our cases corresponding to characteristic classes in $H^k(M;R) \ \cong \ H^k(N(\mathcal{U}); R)$, to construct transition functions $\Omega_{j \ell} \colon U_j \cap U_\ell \to GL(k, \mathbb{R})$ and define $f_\ell \colon U_\ell \to V_k(\mathbb{R}^{J\cdot k})$ by $f_\ell(b) = \left[  \sqrt{\phi_j(b)}\Omega_{j\ell}(b) \right]_{j=1}^J$.

Lastly, this collection of local lifts sending $b \in U_\ell \subset M$ to $f_\ell(b)$ constitutes a one-to-many map $f \colon M \to V_k(\mathbb{R}^{J\cdot k})$. Because $b$ may lie in distinct $U_j$, this map is only well-defined up to right multiplication by elements from $O(k)$ but upon passing to the quotient of $V_k(\mathbb{R}^{J\cdot k})$ by $O(k)$, namely $Gr(k, \mathbb{R}^{J\cdot k})$, one obtains the classifying map for the given bundle. In other words, we have constructed an embedding $f$ of the data such that the following diagram commutes
\begin{align}
\xymatrix{
    M \ar@{.>}[r]^-f \ar[dr]_-g & V_k(\mathbb{R}^{J\cdot k})  \ar[d]\\
    & Gr(k, \mathbb{R}^{J\cdot k}) \ar[r]^\subset & Gr(k, \mathbb{R}^{\infty})
}
\end{align}
where the vertical map is the natural projection given by the quotient of the $O(k)$ action and $g$ is the classifying map of the bundle. We now consider three examples of this data modality.

\subsubsection{The M{\"o}bius bundle}
Let $M = \mathbb{R}/\mathbb{Z} = S^1$ be the circle with circumference 1, $X \subset S^1$ a discrete data set sampled from $M$, and consider the M{\"o}bius bundle over $M$, a rank-1 vector bundle. Let $\mathcal{U} = \sett{U_j}$ be the trivializing open cover of $\mathbb{R}/\mathbb{Z}$ given by open sets of the form $U_j = (\frac{j-1}{25}, \frac{j+1}{25}) \mod 1$ for $1 \leq j \leq 25$ and define a partition of unity $\sett{\phi_j}_{j=1}^{25}$ subordinate to $\mathcal{U}$. In this example, transition functions $\Omega_{j \ell} \colon U_j \cap U_\ell \to O(1) = \sett{\pm 1} \subset GL(1, \mathbb{R})$ are given as $\Omega_{j\ell} = 1$ if $\ell = j+1\mod 25$, $\Omega_{j\ell} = -1$ if $\ell=j-1\mod 25$, and 0 otherwise.

For this series of experiments, we sampled $20$ different sets of $1000$ points drawn from a random uniform distribution in $M = \mathbb{R}/\mathbb{Z}$ to produce a point cloud $X$ comprising $1000$ points in $V_k(\mathbb{R}^{J\cdot k})$ = $V_1(\mathbb{R}^{25})$. Hypothesizing that $X$ should live near a nonlinearly embedded copy of $S^1$ in $V_1(\mathbb{R}^{25})$, we then performed our dimensionality reduction pipeline with $k=1$, $N=25$, and $n=2$. The left-hand panel of Figure \ref{fig_mobius_combined} summarizes the results of 20 trials using 1) a randomly chosen $\alpha$, denoted $\alpha_{rand}$, 2) $\alpha_{PCA}$, and 3) $\alpha_{GD}$, and statistical measurements of the distributions.\\

\indent These data demonstrate a clear improvement in mean squared error between $\alpha_{PCA}$ and $\alpha_{GD}$. This improvement of about 12\% in MSE (average MSE of 1.454 for $\alpha_{PCA}$ versus 1.276 for $\alpha_{GD}$) can have significant consequences in terms of data recovery, as we now explore. Our pipeline can be summarized by the following commutative diagram.

\begin{align}\label{eqn_pipeline_for_vector_bundle_reduction}
    \xymatrix{
    X \subset \mathbb{R}/\mathbb{Z} \ar@{.>}^-f[r] \ar^-g[dr] & V_1(\mathbb{R}^{25}) \ar^-{\mod O(1)}[d]& V_1(\mathbb{R}^2) \ \cong \ \mathbb{R}/\mathbb{Z} \ar_-\alpha[l] \ar^-{\mod O(1)}[d]\\
    & Gr(1, \mathbb{R}^{25}) & Gr(1, \mathbb{R}^2) \ \cong \ \mathbb{R}P^1 \ \cong \ \mathbb{R}/\mathbb{Z} \ar_-\alpha[l]
    }
\end{align}

In this example, the target(s) of the PSC pipeline, namely $V_1(\mathbb{R}^2)$ and $Gr(1, \mathbb{R}^2)$, are isomorphic to the circle $S^1$. Since our data was sampled from $S^1$, we can compare each data point's original ``ground truth'' angular coordinate to its recovered angular coordinate in $V_1(\mathbb{R}^2)$ and $Gr(1, \mathbb{R}^2)$. Note that 1) because $f$ is a one-to-many lift of the continuous image of $S^1$ in $Gr(1, \mathbb{R}^{25})$ we only expect a 1-1 recovery of the original data upon passing to the relevant Grassmannian manifold, and 2) we are able to perform dimensionality reduction to $Gr(1, \mathbb{R}^2)$ thanks to the equivariance of the PSC pipeline. The right-hand panel of Figure \ref{fig_mobius_combined} shows this coordinate recovery for one particular trial, and should be interpreted similarly to Figure \ref{fig_neuro_ex_coord_recovery_with_pipeline} where a line of slope 1 or -1 corresponds to ``perfect'' recovery of the data up to orientation (clockwise or anticlockwise) and a uniform translation. 

\begin{figure}
    \centering
    \includegraphics[width=0.8\linewidth]{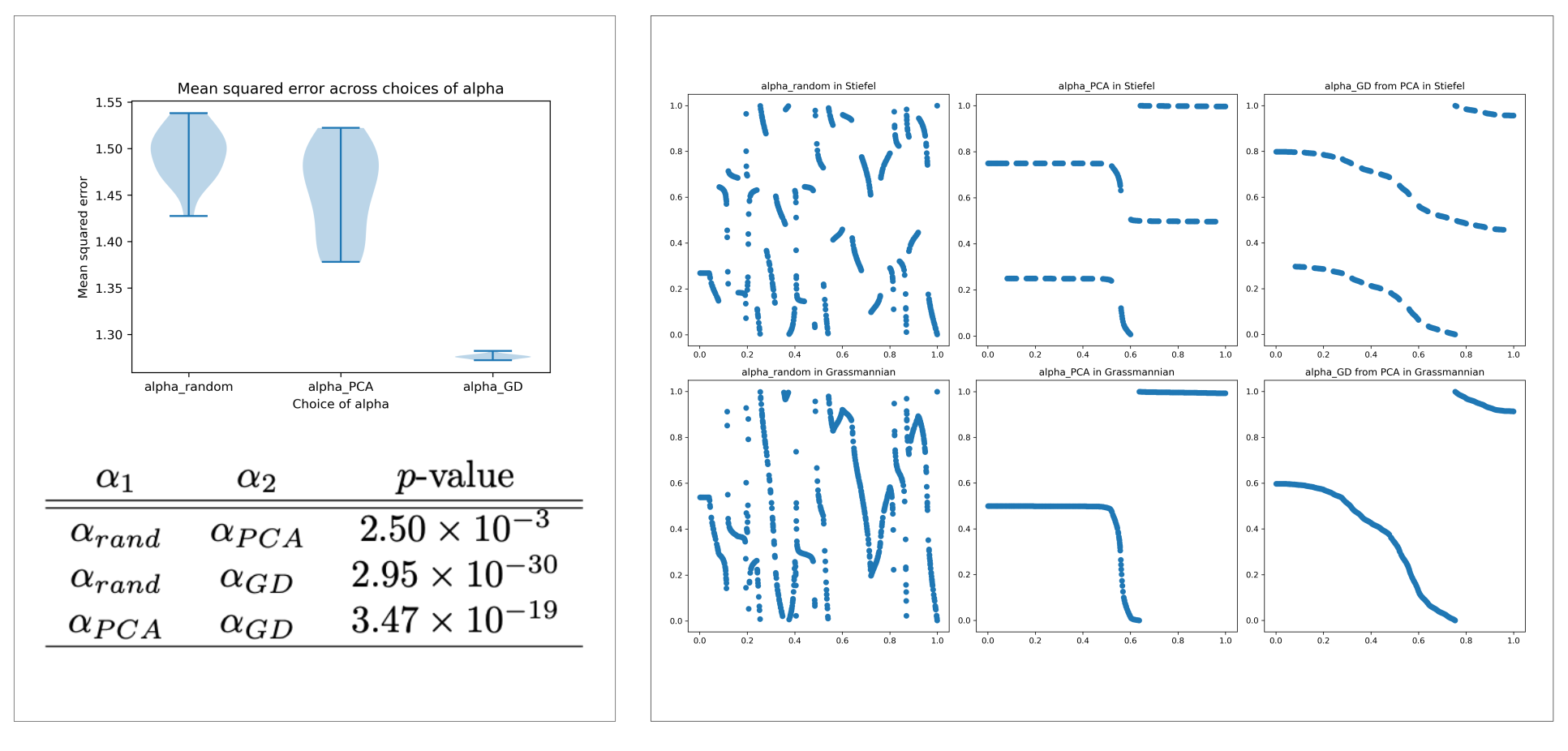}
    \caption{(Left-hand panel, top) Violin plots of MSE of $\pi_\alpha$ for various choices of $\alpha$ across 20 trials. (Right-hand panel, bottom) Table of $p$-values of independent T-tests for all possible pairs of $\alpha$, under the null hypothesis that the MSE scores are drawn from distributions with equal expected value. With a significance threshold of 0.05, we reject the null hypothesis in all cases.(Right-hand panel, top row) Plots of ground truth angular coordinate (x-axis) in $S^1$ against recovered angular coordinate (y-axis) in $V_1(\mathbb{R}^2)$ via the PSC pipeline of a typical data set of 1000 points sampled uniformly from $S^1$, for varying choices of $\alpha$. (Right-hand panel, bottom row) Same as top row, but with y-axis replaced by recovered angular coordinate in $Gr(1, \mathbb{R}^2)$.}
    \label{fig_mobius_combined}
\end{figure}

As is evident from the right-hand panel of Figure \ref{fig_mobius_combined}, coordinate recovery in $Gr(1, \mathbb{R}^2)$ is significantly better (i.e., closer to a line of slope $\pm 1$) for $\alpha_{GD}$ than all other choices of $\alpha$, highlighting that the improvements in MSE on the order of just 10\% obtained by $\alpha_{GD}$ are correlated with more faithful downstream representations of the original data. 

\subsubsection{Whitney sums and the torus}
The above construction for $S^1$ admits a generalization to an infinite family of examples. Let $\tau\to S^1$ be the M\"obius band bundle of the previous section, and let $p_i:S^1\times S^1\to S^1$ be the projection to the $i$th factor. We may form the Whitney sum of $p_1^*(\tau)\oplus p_2^*(\tau)$, a rank-2 bundle. For $\mathcal{U}$ a trivializing open cover for $\tau\to S^1$, the Cartesian product $\mathcal{U}\times \mathcal{U}$ constitutes a trivializing open cover for $p_1^*(\tau)\oplus p_2^*(\tau)$.

It is also useful that such sums (see \cite[Example 10.7]{leesmoothmanifolds}) have transition functions which are diagonal matrices. Let $\Omega^1_{jl}, \Omega^2_{km}\in O(1)$ denote distinct transition functions for $\tau$ on $S^1$, defined on $U_j\cap U_l$ and $U_k\cap U_m$ respectively. Then $\widetilde{\Omega}_{jklm}\in O(2)$ is a transition function on the Whitney sum on the intersection $(U_j\cap U_k)\times (U_l \cap U_m)$. The value on the first entry of the diagonal of $\widetilde{\Omega}_{jklm}\in O(2)$ at some $(x_1, x_2)\in S^1\times S^1$ is $\Omega^1_{jl}(x_1)$, and similarly for the second coordinate. 

Lastly, recall that if $\{\phi_k\}_{k=1}^J$ is a partition of unity on a space $X$, then $\{\phi_k\phi_l\}_{k,l = 1}^J$, consisting of all pairwise products of the $\phi_k$, is a partition of unity on $X\times X$. Thus, the data specified in the previous section for $k=1$ is sufficient to determine a rank 2 vector bundle on the torus $T^2=S^1\times S^1$. We may then attempt dimensionality reduction by mapping to $V_2(\mathbb{R}^n)$ for $n\geq 2$. For the uniform distribution, MSE is significantly decreased when fitting $V_2(\mathbb{R}^3) \cong O(3)$ to $T^2$ rather than $O(2)$, while $O(2)$ is more appropriate for an embedded $(p,q)$ curve as matching dimensions would suggest.

In this series of experiments, we took $J = 100$ open sets in our trivializing cover $\mathcal{U}$ of $T^2$, induced by an open cover of $S^1$ with 10 open sets. We then sampled 1000 points on $T^2$ from four different distributions: 1) the uniform distribution on $T^2$, 2) the uniform distribution on a (1,1)-curve, and 3) the uniform distribution on a (1,15)-curve. Here a $(p,q)$-curve wraps $p$ times about the meridian of the torus and $q$ times about the longitude. 

With these generation methods in place, we perform the PSC pipeline. In the case of the uniform distribution on $T^2$, we project to $V_2(\mathbb{R}^3) \ \cong \ O(3)$, while for $(p,q)$-curves we project to $V_1(\mathbb{R}^2)$ only. In each case, we compare mean squared error for gradient descent against mean squared error for PCA. Figure \ref{fig_violin_and_p_vals_t2} shows the aggregate results for a series of 20 trials of each choice of dimension and data generation method. Each example exhibits significant improvements in MSE obtained by $\alpha_{GD}$. The same analysis (omitted here but using code provided in \texttt{t2rank2embedding.ipynb}) in the case of the trivial bundle $\underline{\mathbb{R}}\oplus\underline{\mathbb{R}}\to T^2$ still exhibits a difference in MSE when comparing PCA to gradient descent. 

\begin{figure}[ht]
    \begin{tabularx}{\linewidth}{@{}c X @{}}
    \includegraphics[width=0.5 \linewidth,valign=c]{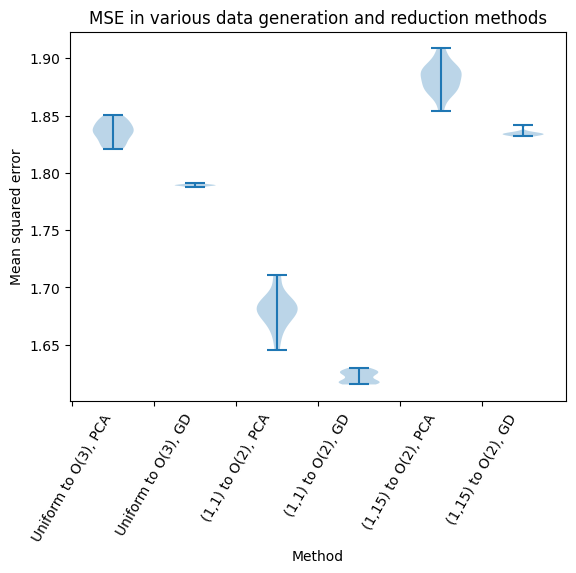}
    & 
\begin{tabular}{@{}ccc@{}}
\toprule
Sample & $n$ in $V_2(\mathbb{R}^n)$ & $p$-value\\
\midrule
Uniform $T^2$ & 3 & $7.81\times 10^{-24}$ \\
$(1,1)$-curve & 2 & $9.24\times 10^{-19}$ \\
$(1,15)$-curve & 2 & $7.54\times 10^{-19}$\\
\bottomrule
\end{tabular}
    \end{tabularx}
\caption{(Left) Violin plots of MSE of $\pi_\alpha$ for various situations across 20 trials. (Right) Table of $p$-values of independent T-tests comparing total MSE using PCA vs total MSE using gradient descent, under the null hypothesis that the MSE scores within each sampling method for each dimensional reduction method are drawn from distributions with equal expected value. At a significance threshold of 0.05, we reject the null hypothesis in all cases. }
\label{fig_violin_and_p_vals_t2}
\end{figure}

\subsubsection{The tangent bundle on \texorpdfstring{$S^2$}{S2}}
Lastly, we consider the (rank-2) tangent bundle $TS^2$ on $S^2$. We begin by constructing an explicit representative of the Euler class of the tangent bundle $TS^2$, which corresponds to $2 \in \mathbb{Z}$ under the identification $H^2(S^2; \mathbb{Z}) \ \cong \ \mathbb{Z}$. Provided we can construct a trivializing open cover $\mathcal{U}$ whose nerve $N(\mathcal{U})$ has no 3-simplices, $\eta$ can be taken to have the value $2$ on an arbitrary 2-simplex of $N(\mathcal{U})$ and the value $0$ on all others. To construct such a $\mathcal{U}$, we sample points $\sett{v_j}_{j=1}^J$ on $S^2$ using the Fibonacci lattice and then construct a triangulation of $S^2$ using the $v_j$'s as vertices. Next, we let $U_j$ be the open star of $v_j$, i.e., the interiors of all simplices of the triangulation that contain $v_j$ as a face, along with $v_j$ itself and let $\mathcal{U} = \sett{U_j}_{j=1}^J$. By construction, the nerve of $\mathcal{U}$ is precisely the triangulation of $S^2$ we began with, so has no 3-simplices.

Next, we replace $\eta \in Z^2(N(\mathcal{U}); \mathbb{Z})$ with its harmonic representative $\theta \in Z^2(N(\mathcal{U}); \mathbb{R})$ in order to produce smoother transition functions $\Omega_{j\ell} \in O(2) \subset GL(2, \mathbb{R})$. For details on this procedure, see \cite[Section 6]{perea2018multiscale}. In brief, let $\iota_\ast \colon Z^2(N(\mathcal{U}); \mathbb{Z}) \to Z^2(N(\mathcal{U}); \mathbb{R})$ be the map induced by the inclusion $\iota \colon \mathbb{Z} \to \mathbb{R}$. Let $\delta \colon C^1(N(\mathcal{U}); \mathbb{R}) \to C^2(N(\mathcal{U}); \mathbb{R})$ be the boundary map of simplicial cochains, $\delta^+$ its Moore-Penrose pseudoinverse, and define $\nu = \delta^+(\eta)$. Then $\theta$ is given by $\theta = \iota_\ast(\eta) - \delta(\nu)$ and transition functions are given by the following, for $b \in U_r \cap U_s$.

\begin{align}
    \Omega_{j\ell}(b) = \begin{bmatrix}
        \cos\left(2\pi(\nu_{j\ell} + \sum_m\phi_m(b)\theta_{mj\ell}) \right) & -\sin\left(2\pi(\nu_{j\ell} + \sum_m\phi_m(b)\theta_{mj\ell}) \right)\\
        \sin\left(2\pi(\nu_{j\ell} + \sum_m\phi_m(b)\theta_{mj\ell}) \right) & \cos\left(2\pi(\nu_{j\ell} + \sum_m\phi_m(b)\theta_{mj\ell}) \right)
    \end{bmatrix}
\end{align}

For this series of experiments, we took $J = 100$ open sets in our trivializing cover $\mathcal{U}$ and sampled 1500 points on $S^2$ from three different distributions: 1) the uniform distribution on $S^2$, 2) Gaussian distributions around two orthogonal great circles on $S^2$, and 3) Gaussian distributions around three different centers on $S^2$ to obtain, in each case, point clouds in $V_2(\mathbb{R}^{100\cdot 2}) = V_2(\mathbb{R}^{200})$. Given that $S^2$ is a 2-dimensional manifold and $V_2(\mathbb{R}^n)$ is of dimension $2n-3$, we chose the least such $n$ for which $2n-3 \geq 2$ for the PSC pipeline, i.e., $n=3$. Figure \ref{fig_violin_and_p_values_sphere_tanget_bundle} shows the aggregate results analogous to Figure \ref{fig_mobius_combined} for a series of 10 trials of each distribution. In all cases, $\alpha_{GD}$ achieves a significantly lower average MSE than $\alpha_{PCA}$.

\begin{figure}
    \centering
    \includegraphics[width = 0.75\textwidth]{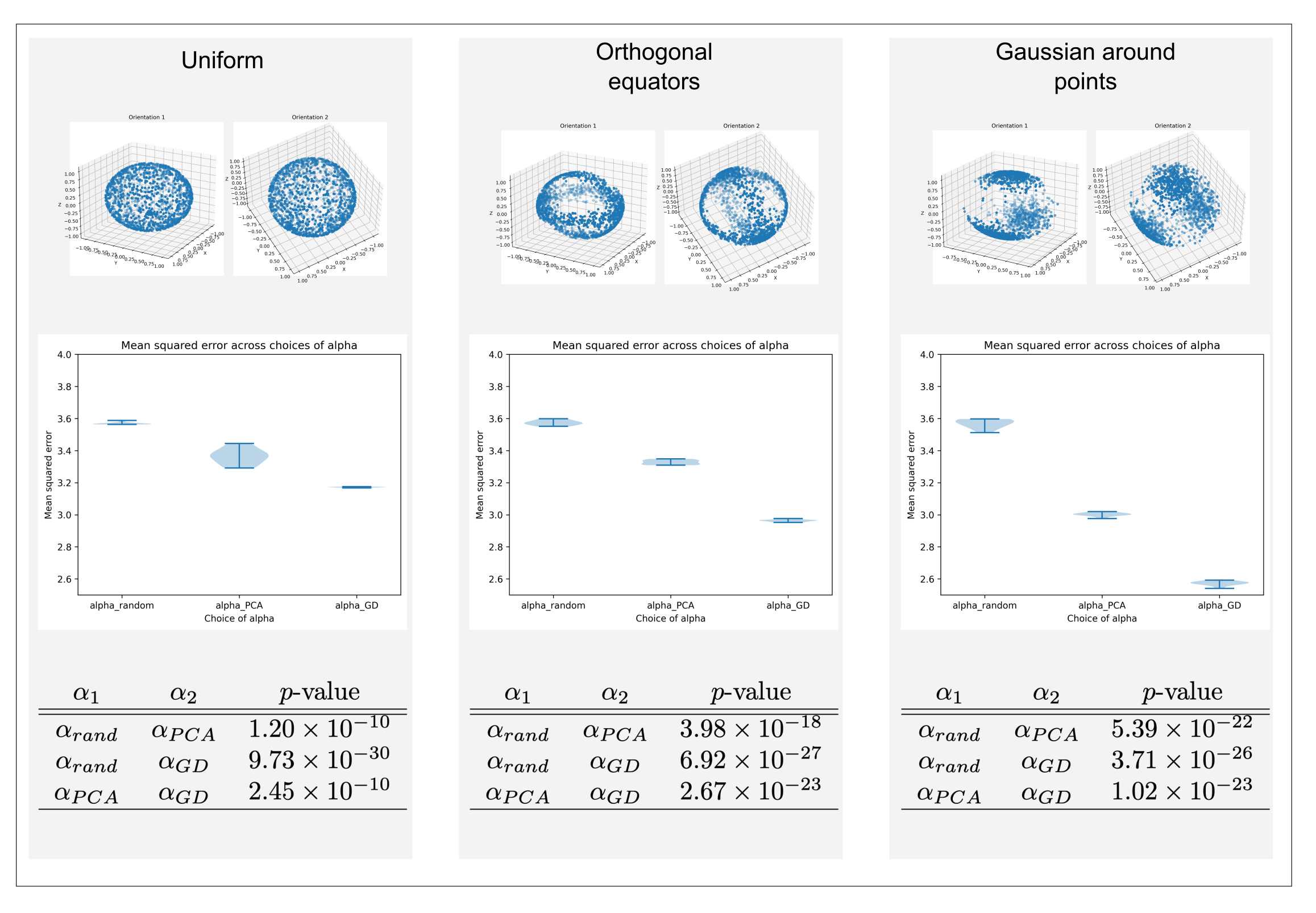}
    \caption{(Top row) Examples of the three distributions in $S^2$ from which data was sampled for this series of experiments. (Middle row) Violin plots of MSE of $\pi_\alpha$ for various choices of $\alpha$ across 10 trials, for each distribution. (Bottom row) Tables of $p$-values of independent T-test for all possible pairs of $\alpha$, under the null hypothesis that the MSE scores are drawn from distributions with equal expected value. With a significance threshold of 0.05, we reject the null hypothesis in all cases.}
    \label{fig_violin_and_p_values_sphere_tanget_bundle}
\end{figure}

\section*{Acknowledgments}

This material is based upon AMS Mathematics Research Communities (MRC) work supported by the National Science Foundation under Grant Number DMS-1916439. 
H. Lee's work was partially completed while the author was at Department of Mathematics, University of California Los Angeles, and was partially supported by NSF DMS-1952339. 
J. A. Perea's work was partially supported by the National Science Foundation through grants  CCF-2006661, and CAREER award DMS-2415445. 
N. Schonsheck's work is supported by the Air Force Office of Scientific Research under award FA9550-21-1-0266. This author also wishes to thank Chad Giusti for his continued support.

\bibliographystyle{siamplain}
\bibliography{main}

\appendix

\section{RANSAC}
This is an optional add-on to the initialization described in PSC meant to help in handling outliers. That is, for real-world cases in the presence of noise and where the optimal value of $n$ is unknown, there may be some data points that lie far from the image $\alpha(V_k(\mathbb{R}^n))$ of a low-dimensional Stiefel manifold, even if the majority of the data points do live near $\im(\alpha)$. Since PCA is not robust to outliers, they may limit the efficacy of the initialization $\alpha_{PCA}$. To address this, the algorithm below implements a random sample consensus (RANSAC)-type initialization adapted to the PSC pipeline. While we have not encountered cases where this additional step is needed, \cite{ransac_fischler_bolles} and \cite{ransac_choi_kim_yu} demonstrate the potential utility of such an approach and we have included it for the sake of completeness and user-friendliness. (See also \cite{ransac_pca} for an implementation of RANSAC to construct a robust PCA.) We suppose as given a data set $\mathcal{Y} \subset V_k(\mathbb{R}^N)$ and fix a user-chosen target dimension $n \ll N$ for the dimensionality reduction.

\medskip

\begin{breakablealgorithm}\label{alg_ransac_initialization}
    \caption{RANSAC initialization for the PSC pipeline}
    \begin{algorithmic}[1]
        \STATE Fix a user-defined parameter $p$ and outlier threshold $\tau$, e.g. $p = 99\%$ and $\tau=3$.
        \STATE Uniformly randomly sample $p$ of the initial data points $\mathcal{Y}$ to obtain $\mathcal{Y}'$.
        \STATE Construct $\alpha_{PCA}$ on the set $\mathcal{Y}'$ using Steps 1-4 of PSC.
        \STATE Compute the distance $\|y - \pi_{\alpha_{PCA}}(y)\|_\mathrm{F}$ for each point $y \in \mathcal{Y}'$. Compute mean and standard deviation among those distances. Let $\mathcal{B}$ be the set of points $y$ for which the distance is at least $\tau$ standard deviations away from the mean.
        \STATE Repeat Steps 2-4 with $\mathcal{Y} = \mathcal{Y} - \mathcal{B}$ until Step 4 terminates with $\mathcal{B} = \emptyset$.
        \STATE Use the resulting map $\alpha_{PCA}$ as the initialization for PSC.
    \end{algorithmic}
\end{breakablealgorithm}
\medskip

Depending on user choices, it is possible that Algorithm \ref{alg_ransac_initialization} will remove all data points from consideration in that $\mathcal{Y} - \mathcal{B}$ could be empty. Such a case would suggest that the data does not lie near $\alpha(V_k(\mathbb{R}^n))$, and so $n$ should be larger, or PSC is not an appropriate method with which to analyze the data.

\section{Additional figures for stimulus space model experiment} \label{sec:appendix-stimulus-space-model}

We tried to provide a fair comparison between different methods by applying similar post-processing steps. We first fixed discontinuities based on visual inspection of Figure \ref{fig:comparisons_others_path}. For MDS, $2\pi$ was subtracted from all points above 5.4. For persistent cohomology and $\alpha_{PCA}$, $2\pi$ was added to all points below 1. Afterwards, the same rotating, scaling, and Gaussian smoothing with $\sigma=100$ was applied, yielding Figure \ref{fig:comparison-others} below.

\begin{figure}[ht]
\centering
\includegraphics[width=0.5\linewidth]{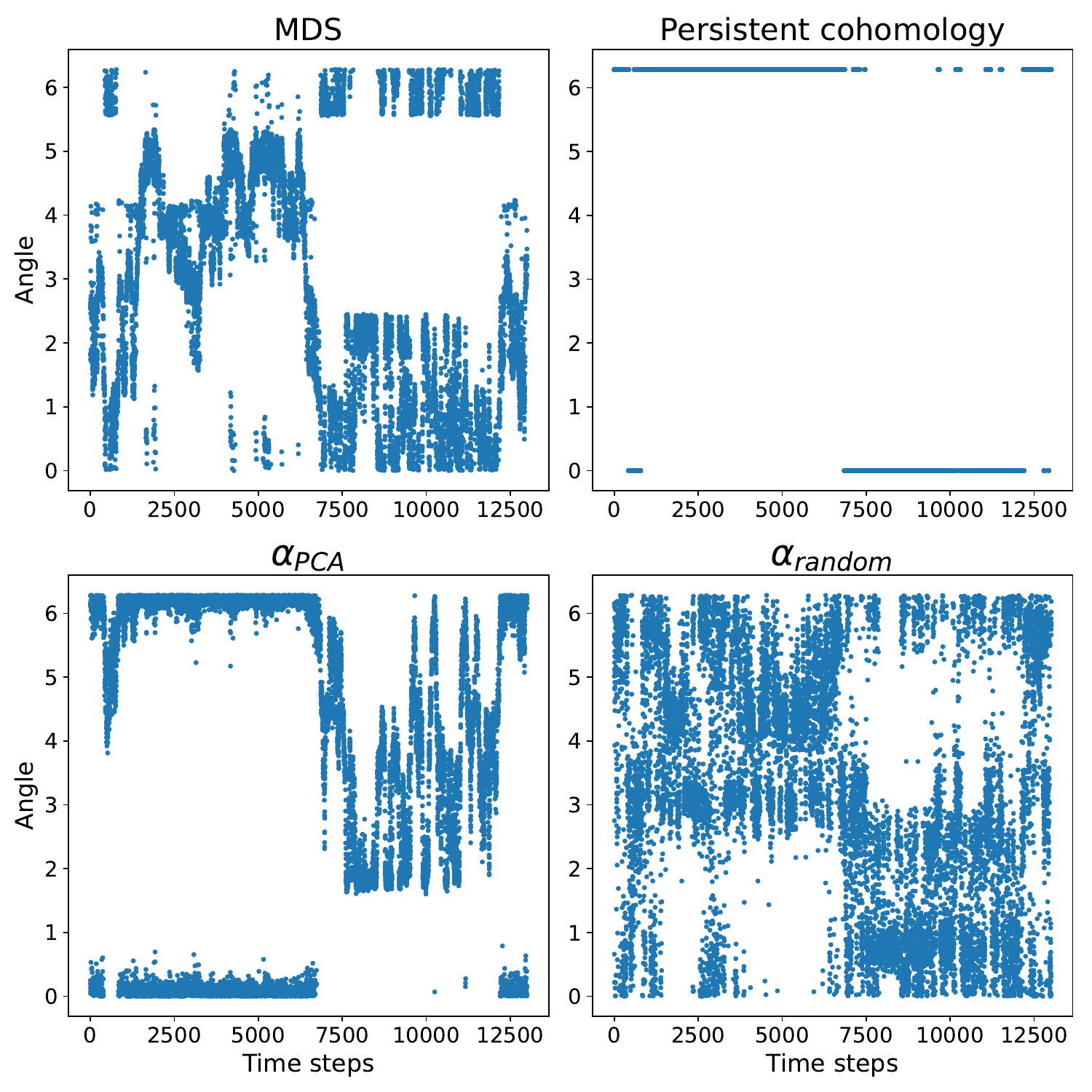}
\caption{Stimulus space model experiment results from multi-dimensional scaling, Persistent cohomology, $\alpha_{PCA}$ and $\alpha_{random}$ without any post-processing or smoothing.}
\label{fig:comparisons_others_path}
\end{figure}

\begin{figure}
    \centering
    \includegraphics[width=0.45\linewidth]{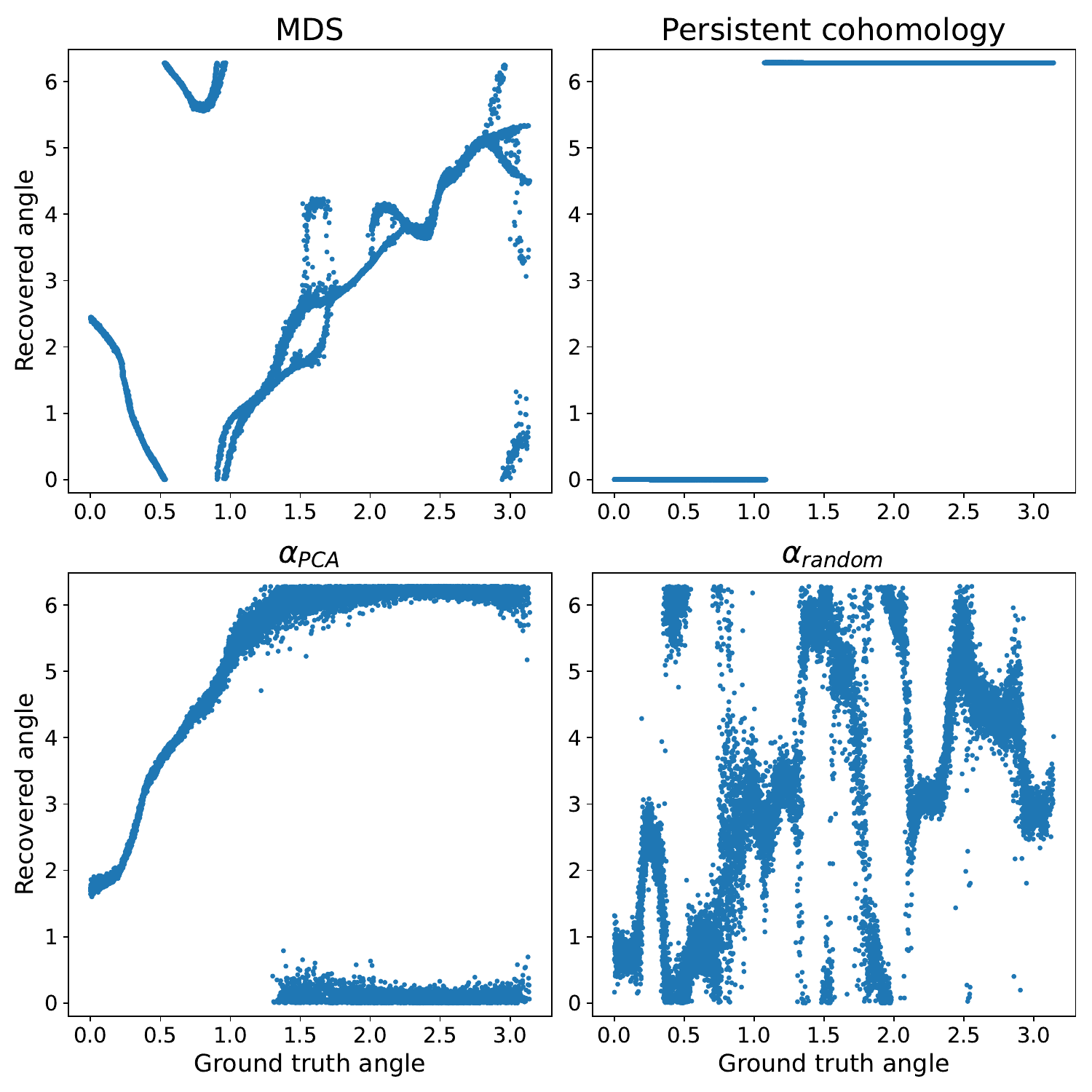}\hfill  
    \includegraphics[width=0.45\linewidth]{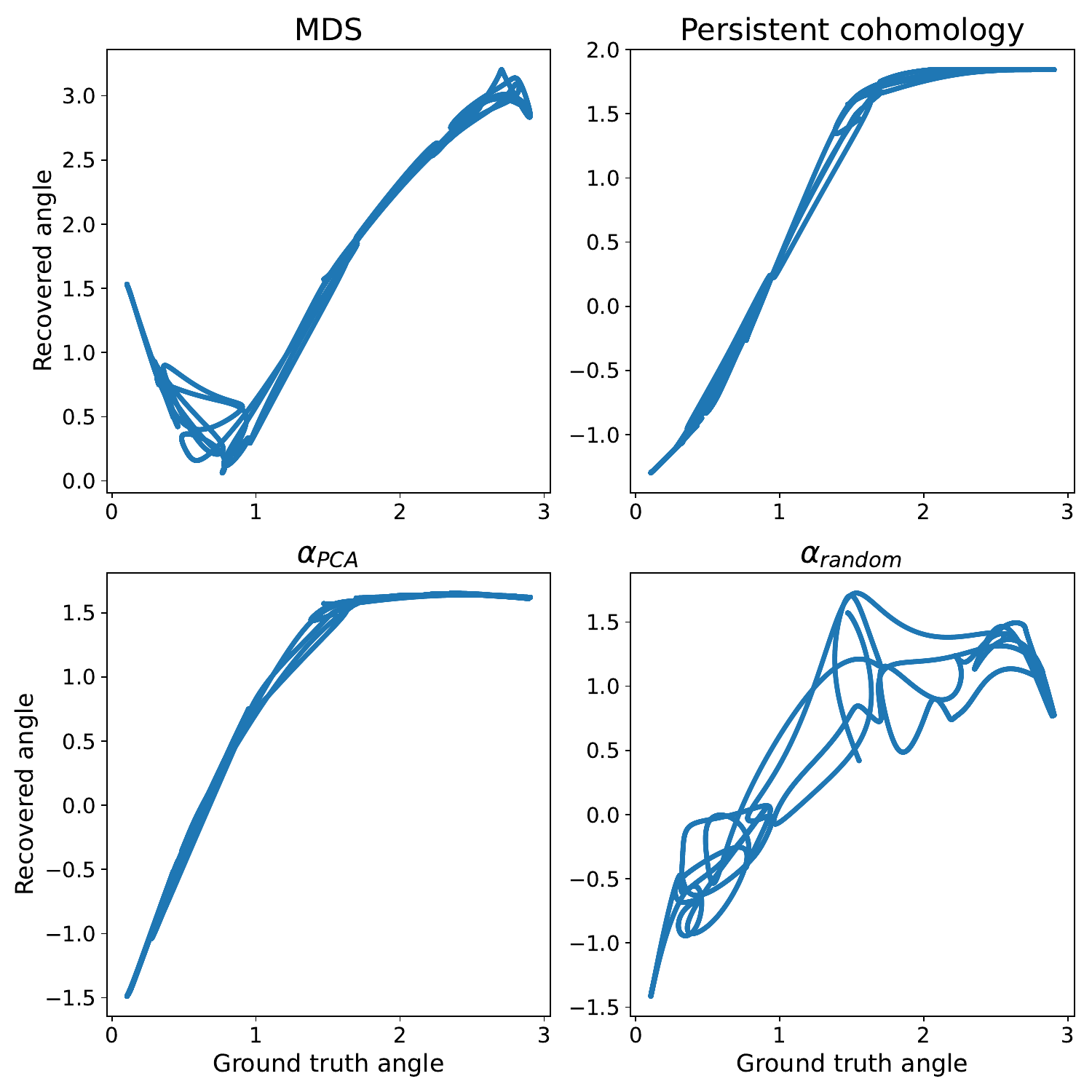}
    \caption{Left: Angular coordinates of a geometric walk on a half-circle against recovered coordinates. Middle: Post-processed and smoothed recovered coordinates against smoothed ground truth angles.}
    \label{fig:comparison-others}
\end{figure}

\section{Efficient implementation}
We made considerable effort to speed up the implementation of our method.

\subsection{Checking rank condition}

The following provides a description of the domain $T_\alpha$ that is amenable to direct computation.
\begin{proposition}\label{prop_alt_description_of_neighborhood}
For a fixed $\alpha \in V_n(\mathbb{R}^N)$ and $y \in V_k(\mathbb{R}^N)$, let $A = \alpha^\top y$. Then $\rank(\alpha^\top y) = k$ if and only if $\mathrm{det}(A^\top A) \neq 0$.    
\end{proposition}
\begin{proof}
    We know that $\rank(A)$ is equal to the number of nonzero singular values of $A$, which in turn is equal to the number of nonzero eigenvalues of $A^\top A$. Since $A^\top A$ is symmetric, it follows from the Spectral Theorem that it has $k$ (not necessarily distinct) eigenvalues $\lambda_1, \lambda_2, \ldots \lambda_k$. Hence, $\rank(A) = k$ if and only if none of the $\lambda_i$ are zero. The Proposition follows from the fact that the determinant of $A^\top A$ is equal to $\prod_{i=1}^k\lambda_i$.
\end{proof}

Given a data set $\mathcal{Y} \subset V_k(\mathbb{R}^N)$ and fixed $\alpha \in V_n(\mathbb{R}^n)$, to determine if any of the points $y \in \mathcal{Y}$ lie outside of the domain, one can precompute the matrix product $\alpha \alpha^\top$ and then iteratively compute the determinant of $y^\top (\alpha \alpha^\top) y$. In practice, we find that such a computation is faster than an explicit rank computation in some cases; see Figure \ref{fig:runtime}. We varied the number of points as well as the ambient space dimension $N$ between [10, 50, 100, 500, 1000]. $n$ was fixed at $0.2N$ and $k=0.1N$. We checked for rank condition using both rank of $\alpha^\top y$ and determinant of $y^\top\alpha\alpha^\top y$, and also computed $\alpha_{GD}$ using the same data (Figure \ref{fig:runtime_psc}). This set of experiments were done on a standard MacBook Pro 2023 with 16GB memory. 

\begin{figure}[htp]
    \centering
    \includegraphics[width=0.3\linewidth]{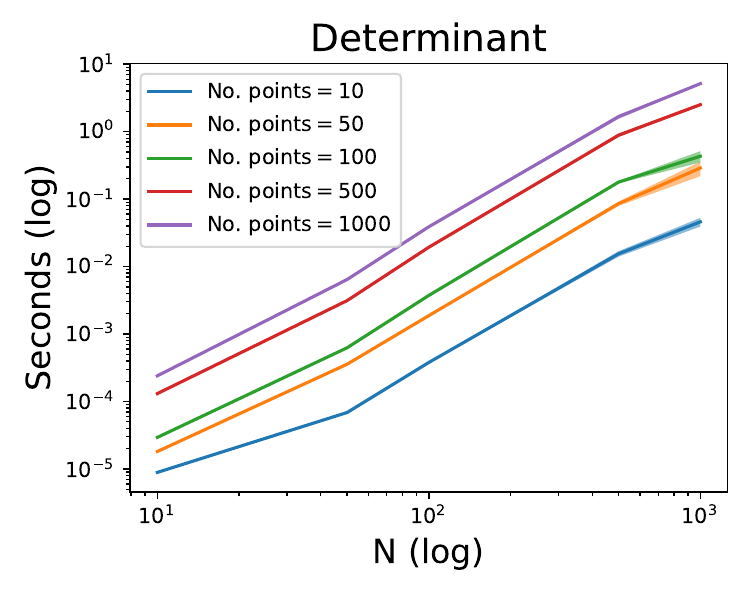}%
    \includegraphics[width=0.3\linewidth]{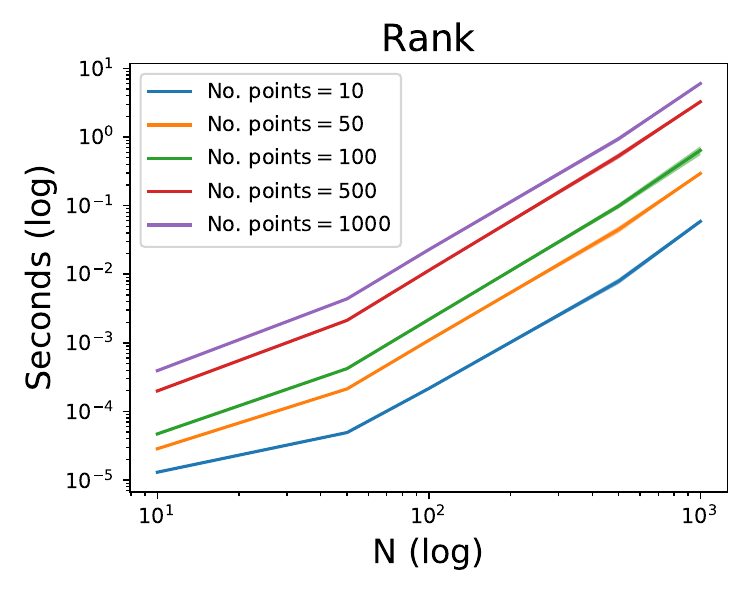}
    
    \includegraphics[width=0.3\linewidth]{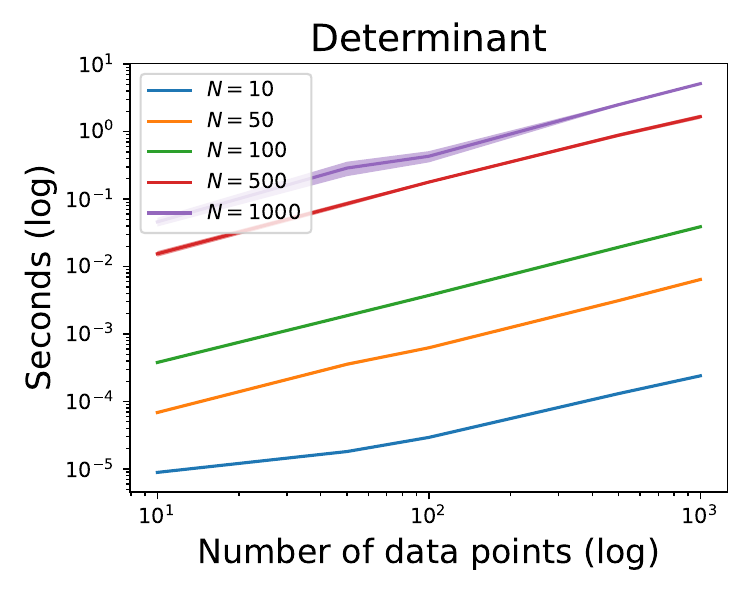}%
    \includegraphics[width=0.3\linewidth]{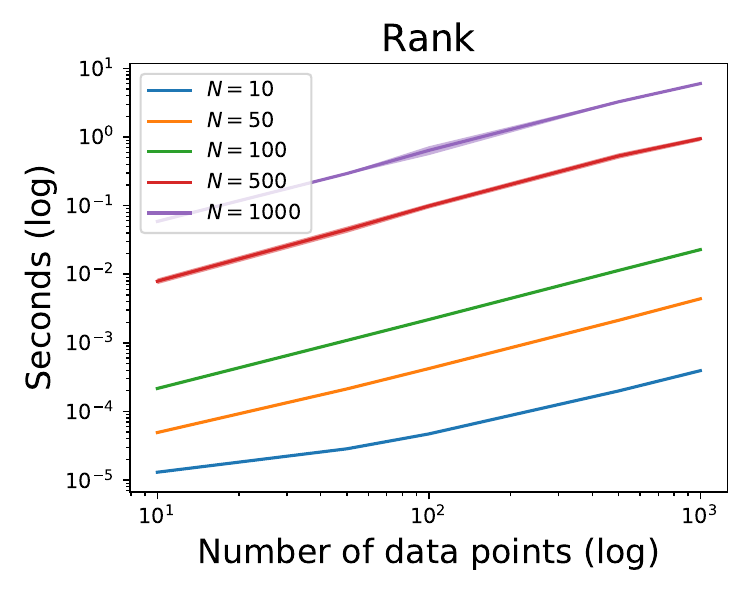}
    
    \caption{Runtime of checking rank condition. Mean and standard deviation calculated over 5 runs.}
    \label{fig:runtime}
\end{figure}

\subsection{\texorpdfstring{$\alpha_{PCA}$}{alpha PCA}}
A faster and more memory-efficient way to derive $\alpha_{PCA}$ is to concatenate the $y$'s horizontally to get a matrix of size $N\times k|\mathcal{Y}|$. We then apply singular value decomposition, select $n$ of the left singular vectors that correspond to the $n$ highest singular values, and set them to be columns of $\alpha_{PCA}$.

\subsection{Utilizing Python numpy}

\begin{figure}
    \centering
   \includegraphics[width=0.3\linewidth]{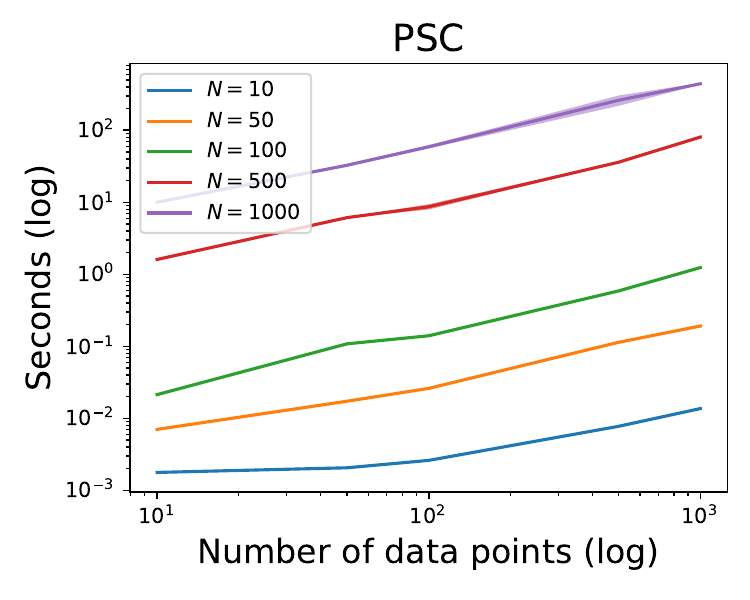}
      \includegraphics[width=0.3\linewidth]{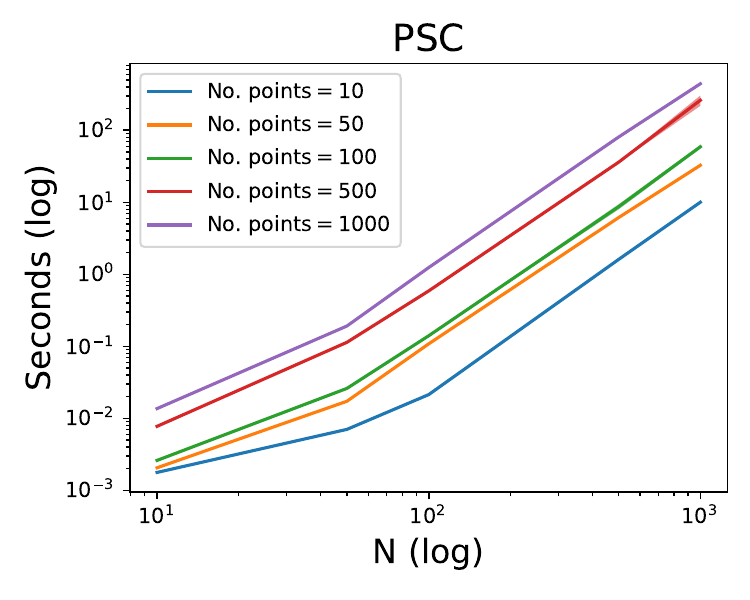}
      \caption{Runtime of calculating $\alpha_{GD}$. Mean and standard deviation calculated over 3 runs.}
    \label{fig:runtime_psc}
\end{figure}
We are able to avoid many for-loops using vectorized operations in numpy. Since most speed-ups happen behind the scenes of numpy and pymanopt, it is difficult to provide a detailed theoretical analysis. However, for reference, we share in Table \ref{tab:runtime} how long each calculation took using the stimulus space model data, where $|\mathcal{Y}| = 13,000, N=100, n=2, k=1$. See notebook \texttt{test\_faster\_PSC.ipynb} for more details, including hardware information.

\begin{table}[h]
    \centering
    \footnotesize
    \begin{tabular}{ccccc}
    \toprule
         & $\alpha_{PCA}$ & $\alpha_{GD}$ &$\pi_{\alpha}$ & $\hat{y}_{\alpha}$\\\midrule
       Naive  & 1.85 s ± 476 ms &345.68 seconds&276 ms ± 23 ms&313 ms ± 81.3 ms\\
       Efficient  & 1.17 s ± 324 ms &8.93 seconds&43.7 ms ± 748 µs&27.8 ms ± 898 µs\\\midrule
      Speedup &36\%&97\%&84\%&91\%\\\bottomrule
    \end{tabular}
    \caption{How long each calculation takes using the stimulus space model data, where $|\mathcal{Y}| = 13,000, N=100, n=2, k=1$. Averaged over 7 runs.}
    \label{tab:runtime}
\end{table}

\end{document}